\pdfoutput=1

\documentclass[11pt]{article}
\usepackage[a4paper, total={16cm, 24cm}]{geometry}

\usepackage[table]{xcolor}
\usepackage{natbib}
\renewcommand\tableofcontents{\listoftoc*{toc}} 

\title{Discovering Consistent Subelections}

\usepackage{authblk}
\author[1]{Łukasz Janeczko}
\author[2]{Jérôme Lang}
\author[1]{Grzegorz Lisowski}
\author[1,2]{Stanisław Szufa}
\affil[1]{AGH University, Poland} 
\affil[2]{CNRS, LAMSADE, Université Paris Dauphine-PSL, France}

\usepackage{caption} 
\usepackage{algorithm}
\usepackage{algorithmic}

\usepackage{graphicx} 
\usepackage{microtype}
\usepackage{amsmath}
\usepackage{nicefrac}
\usepackage[textsize=tiny]{todonotes}
\usepackage{paralist}
\usepackage{thm-restate}
\usepackage{times}
\usepackage{soul}
\usepackage{url}
\usepackage{amsmath}
\usepackage{amsthm}
\usepackage{booktabs}
\usepackage{algorithm}
\usepackage{algorithmic}
\usepackage[switch]{lineno}
\usepackage{collcell}
\usepackage{amsfonts}
\usepackage{mathdots}
\usepackage{subcaption}
\usepackage{nicefrac}
\usepackage{algorithm}
\usepackage{algorithmic}
\usepackage{endnotes}
\usepackage{amssymb}
\usepackage{verbatim}
\usepackage{subcaption}

\usepackage[pagebackref]{hyperref}
\hypersetup{
		pdfencoding=auto, 
		psdextra,
		colorlinks=true,
		citecolor=green!40!black,
		linkcolor=red!50!black,
		urlcolor=blue!80!black
	}

\usepackage[nameinlink]{cleveref}

\usepackage{newfloat}
\usepackage{listings}
\DeclareCaptionStyle{ruled}{labelfont=normalfont,labelsep=colon,strut=off} 
\floatstyle{ruled}
\newfloat{listing}{tb}{lst}{}
\floatname{listing}{Listing}

\newtheorem{theorem}{Theorem}
\newtheorem{proposition}[theorem]{Proposition}
\newtheorem{definition}{Definition}

\newtheorem{corollary}[theorem]{Corollary}
\newtheorem{remark}{Remark}
\newtheorem{lemma}{Lemma}
\newtheorem{observation}{Observation}

\newcommand{\np}{\textsc{NP}}

\newcommand{\fpt}{\textsc{FPT}}

\newcommand{\xp}{\textsc{XP}}

\newcommand{\threesat}{\textsc{$3$-SAT}}

\newcommand{\hiddenclones}{\textsc{Hidden-Clones}}
\newcommand{\hiddenidentity}{\textsc{Hidden-ID}}
\newcommand{\hiddenantagonism}{\textsc{Hidden-AN}}

\newcommand{\hiddensp}{\textsc{Hidden-Single-Peaked}}
\newcommand{\sharphiddenclones}{\textsc{\#-Hidden-Clones}}
\newcommand{\sharphiddenidentity}{\textsc{\#-Hidden-ID}}

\newcommand{\maxid}{\textsc{Max-ID}}
\newcommand{\maxan}{\textsc{Max-AN}}
\newcommand{\maxclone}{\textsc{MaxClone}}

\newcommand{\singlepeaked}{\textrm{single-peaked}}
\newcommand{\singlecrossing}{\textrm{single-crossing}}
\newcommand{\groupseparable}{\textrm{group-separable}}

\newcommand{\UN}{{{\mathrm{UN}}}}
\newcommand{\AN}{{{\mathrm{AN}}}}
\newcommand{\ID}{{{\mathrm{ID}}}}
\newcommand{\ST}{{{\mathrm{ST}}}}

\usepackage[textsize=tiny]{todonotes}
\newcommand{\ssnote}[1]{\todo[color=green!30, inline]{Staś: #1}}

\newcommand{\jlnote}[1]{\todo[color=blue!30, inline]{JL: #1}}

\usepackage{booktabs}
\usepackage{float}
\usepackage{ dsfont }
\newtheorem{example}{Example}
\usepackage{cleveref}

\usepackage{enumitem}
\tikzset{fontscale/.style = {font=\relsize{#1}}
}
\usepackage{relsize}
\date{\vspace{-5ex}}

\begin{document}

\maketitle

\begin{abstract}
We show how hidden interesting subelections can be discovered in ordinal elections. An interesting subelection consists of a reasonably large set of voters and a reasonably large set of candidates such that the former have a consistent opinion about the latter. Consistency may take various forms but we focus on three: Identity (all selected voters rank all selected candidates the same way), antagonism (half of the selected voters rank candidates in some order and the other half in the reverse order), and clones (all selected voters rank all selected candidates contiguously in the original election). We first study the computation of such hidden subelections. 
Second, we analyze synthetic and real-life data, and find that identifying hidden consistent subelections allows us to uncover some relevant concepts.
\end{abstract}

\section{Introduction}

Ordinal voting consists in taking as input a preference profile (a collection of rankings over candidates) and producing a winner, or a set of winners, or a collective ranking as output --- this could be called the ``mechanism'' view of voting. A much less studied question consists in discovering, from a preference profile, some hidden properties of the domain at hand. An example of such a study is the discovery of structure among the set of candidates, such as an order of candidates that makes the profile single-peaked (perhaps only approximately), or among the set of voters, such as an order that makes the profile single-crossing. 

There is however more to discover from a preference profile, by focusing on voters and candidates {\em simultaneously}. Imagine an explorer from another planet visiting us and observing a local voting profile over food items. They do not know anything about our food, they do not have the concepts of meat, fish, vegetables, sweet, or spicy. 
They do not have either the concepts of children, or vegetarians, and ignore   
our local cultures. Still, they can observe that a significant group of voters consistently prefer some items to others (say, tofu and lentils to eggs, eggs to fish, and fish to meat) and that a significant group of voters are indifferent to all food items that they do not know as it is not part of their culture.

In a more realistic context, the voters are citizens of a country with unknown world views, and the candidates are political issues. Still in another context, the explorer is a manufacturing company, voters are potential consumers, and candidates are items they would be interested to purchase if they were on the market. 

The hidden information we seek is a 
subset of voters $V'$ and a subset of candidates $C'$ such that voters in $V'$ have {\em consistent} preferences over $V'$, i.e., it is a {\em consistent subelection} of the original election. In order for a subelection to make us learn something interesting, the consistency property has to be meaningful. We focus on the following three consistency properties: 
\begin{description}[left=0pt .. 1ex, itemindent=!]
    \item [Identity:] All the candidates are ranked the same way by all voters (children prefer coke to juice and juice to coffee). 
    \item [Clone structure]: Voters in $V'$ rank all candidates in $C'$ contiguously in the original election (people living outside of Europe are indifferent between non-exported varieties of European cheese;  this is not to say they rank them at the bottom, as they may very well rank them above items that they know they do not like).
    \item [Antagonism]: Half of the voters in $V'$ rank the candidates in $C'$ in the same order, while the others rank them in the opposite one.
\end{description}

Identity is probably the most interesting consistency property. It allows us to discover significant segments of the population with identical preferences on a large fraction of options (e.g., it allows us to discover that a homogeneous set of voters, whatever we want to call it, prefers tofu to eggs, eggs to fish, and fish to meat).
Clone structures also lead to major findings: A subpopulation considering a set of alternatives as clones usually means that they cannot distinguish between them (e.g., uncommon cheese varieties). We include antagonism as it is the natural opposite of identity and it helps us to discover what divides a subpopulation. 
We leave other meaningful properties, such as single-peakedness, single-crossingness or group separability 
for further research.

Note an important difference between, on the one hand, identity and antagonism, and on the other, clone structures: We do not need the original election to check if a subelection satisfies identity or single-peakedness, whereas being a clone structure can only be defined with respect to the original election, which is not the case for identity or antagonism.

In order for a consistent subelection to be meaningful, not only should the property make sense, but the size of the subelection should also be reasonably large (for instance, knowing that two voters out of ten are consistent over three candidates out of ten does not tell us anything interesting). This motivates searching for consistent subelections whose number of voters (resp. candidates) reach a given threshold. 

\begin{example}\label{running}
To clarify the concepts of hidden clones, identity, and antagonism, consider the following election, with six voters expressing  preferences over six candidates: 
\begin{align*}
v_1&: a \succ b \succ c \succ f \succ e \succ d \\
v_2&: c \succ b \succ a \succ d \succ e \succ f \\ 
v_3&: a \succ f \succ e \succ b \succ c \succ d \\
v_4&: d \succ e \succ f \succ c \succ b \succ a\\
v_5&: a \succ c \succ b \succ d \succ e \succ f \\
v_6&: f \succ e \succ a \succ b \succ c \succ  d \\
\end{align*}
We discover some interesting patterns:
\begin{itemize}
    \item All voters agree that candidates $e$ and $f$ are clones;
    and all voters except $v_3$ agree that $a$, $b$, and $c$ are clones: We have a clone set of two candidates for six voters, and a clone set of three candidates for five voters.
    \item $v_1$, $v_3$ and $v_6$ agree on the ranking $a \succ b \succ c \succ d$: This is a hidden identity with three voters and four candidates. All voters except $v_2$ and $v_4$  agree on $a \succ b \succ d$ and on $a \succ c \succ d$: These are hidden identities with four voters and three candidates.
    \item Candidates $d,e$, and $f$ are antagonizing all voters: Half of them ($v_2$, $v_4$, and $v_5$) rank them $d \succ e \succ f$, while the other half ($v_1$, $v_3$, and $v_6$) rank them in the reverse order $f \succ e \succ d$: There is an antagonism for three candidates and six voters.
\end{itemize}

\end{example}

Discovering consistent, large enough subelections not only helps to understand a population better ({\em e.g.}, with respect to food preferences, customer behavior, or political opinions over issues), but it can also benefit a variety of tasks. Notably, if we start eliciting the preferences of a new voter (outside of the original profile), we might discover that they probably belong to some group with consistent preferences, which eases and speeds up the elicitation process. For instance, once discovered that Ann is a vegetarian, asking her if she prefers pork to asparagus is a loss of time.

For the sake of simplicity, in this paper, we choose to focus on subelections that satisfy {\em exact} consistency properties.  Occasionally, if applicable, we relate to the question of finding the closest subelection in terms of swap distance, that is, the minimum number of swaps on adjacent candidates we need to perform to obtain a subelection satisfying some consistency property \cite{fal-kac-sor-szu-was:c:microscope}. While allowing approximate properties (e.g., find a subelection where all voters order almost all candidates in the same way) is definitely interesting, and of course, would allow us to discover larger meaningful subelections, we leave it for further research.

\paragraph{Our Contribution.} We provide an analysis of finding hidden subelections, both from theoretical and empirical perspectives. First, we focus on their computation. We obtain hardness results for finding a sufficiently large hidden identity or antagonism, which are tempered by parameterized tractability results for the number of candidates or voters, and by a translation into an ILP. We note that we obtain a reduced run time of our FPT algorithms due to a graph representation of unanimous preference orders. As to hidden clones, they can be identified in polynomial time.

Then, we perform an empirical analysis, using both synthetic and real-life data.
To analyze the results for synthetic data we use the \emph{map of elections} framework. As to the real-life data, we study a well-known dataset containing preferences over different types of sushi, as well as a political election dataset (obtained through polls) over candidates from the 2014 French presidential election.


\paragraph{Related Literature.}

Discovering structure in elections has received significant attention in the context of {\em single-peakedness}: Given a profile, is it single-peaked with respect to some hidden axis (such as political left-right) on which the candidates are positioned? Since the plausibility of a positive answer quickly decreases with the size of the profile, most of the focus has been laid on approximate single-peakedness; several measures of single-peakedness have been defined,  and the key question is to identify an axis for which the profile maximizes the single-peakedness degree. Variants have been considered (such as Euclidean preferences, single-peakedness on a tree or a circle), as well as other structures (mostly single-crossingness, where the hidden axis bears on the set of voters). The most recent review of work on this trend is \cite{ElkindlacknerPeters22}. 

Uncovering a hidden axis and maximally explaining the preference profile allows us to discover hidden properties of the domain at hand, which is also our motivation. However, uncovering an axis allows us to learn structure {\em of the set of candidates only}. Symmetrically, finding an axis over the set of voters making the profile as much {\em single-crossing} as possible leads us to learning structure {\em of the set of voters only}.

\citet{ElkindFS12} identify {\em clones} in elections (sets of candidates that are ranked contiguously by all voters). 
Part of our contribution generalizes the discovery of clone sets by considering sets of candidates that are considered clones by only some of the voters.\footnote{A specific notion of clone structure occurs when the set of
candidates can be partitioned in two classes, such that each voter
prefers all candidates in one class to all those in the other one
({\em group separability}); see Sections 3.11 and 4.7 of (Elkind, Lackner, and Peters 2022).}

\citet{ColleyGHMN23} identify, from a preference profile, the  {\em divisiveness} of each candidate, measuring the disagreement of the population about it. Various notions of degrees of consensus, conflict, or diversity within an electorate have also been studied \cite{Alcalde-UnzuV13,AlcantudCC13,HashemiE14}. While (some of) the subelections we discover also tell us something about the degree of consensus or conflict in a society, they tell us much more, by localizing the candidates {\em and} the voters on which there is a high consensus or conflict. 

\citet{fal-sor-szu:c:subelection-isomorphism} study the problem of subelection isomorphism, where they analyze the complexity of verifying whether one election is isomorphic to a subelection of the other election. This approach differs from ours because we always focus solely on the inner structure of a single election, hence, we do not have to deal with all the problems related to matching the candidates and voters from different elections.

{\em Biclustering} \cite{GOVA14L}, also known as co-clustering or block-clustering, aims at learning structure in a real-valued, two-dimensional matrix by simultaneously finding a set of rows and a set of columns with similar behavior, i.e., near-identical rows, near-identical columns, rows or columns roughly obtained from each other by an additive or multiplicative factor. This resembles our subelection discovery tasks, with a major difference: Biclustering algorithms work on a cardinal input, while ours is ordinal. This is more important than it may appear. Crucially, expressing a profile as a matrix whose cell corresponding to voter $i$ and candidate $c_j$ is $c_j$'s rank in $i$'s ranking would not help, as our consistency notions cannot be expressed by near-identity or near-linear relations between rows or columns. 

\section{Consistent Subelections}\label{sec:prem}

Let us introduce the basic notions which we use in our analysis. 
For a natural number $t > 0$, we denote as $[t]$ the set $ \{1, \dots, t \}$.

\paragraph{Elections.} An \emph{election} $E=(C,V)$ consists of the set $C$ of \emph{candidates} and the set $V$ of \emph{voters}. We assume that each voter $v$ submits a \emph{ranking} $\succ_v$ over $C$. Also, The {\em size} of an election $E$, denoted $size(E)$, is a pair of integers $(|C|, |V|)$.

\paragraph{Subelections.} For an election $E=(C,V)$, a subset of candidates $C' \subseteq C$, and a subset of voters $V' \subseteq V$, we say that  $E'=(C',V')$ is a \emph{subelection} of $E$ 
for every voter $v \in V'$ and any pair of candidates $c_1, c_2 \in C'$, $c_1 \succ_v c_2$ in $E'$ if and only if $c_1 \succ_v c_2$ in $E$. 




%
\paragraph{Identity.} We say that $E=(C,V)$ is an \emph{identity election} if, for every pair of voters $v_i, v_j \in V$, $\succ_{v_i} \ = \ \succ_{v_j}$. 


\paragraph{Antagonism.} We say that $E=(C, V)$, with $|V|$ even, is an \emph{antagonism} if there is a partition of $V$ into $V_1$ and $V_2$  such that $|V_1|=|V_2|$ and for every pair of voters $v_1 \in V_1$, $v_2 \in V_2$, as well as any pair of candidates $c, c' \in C$, $c \succ_{v_1} c'$ if and only if $c' \succ_{v_2} c$.

\paragraph{Clones.} 
We say that $C'$ is a clone set for an election $E=(C, V)$ if for any $x, y \in C'$, any $z \notin C'$ and any voter $i$ we have $x \succ_i z$ if and only if $y \succ_i z$. The {\em size} of a clone set $C'$ for election $E$ is the pair of integers $(|C'|, |V|)$.

In our initial example, $(\{a,b,c,d\}, \{v_1, v_3, v_6\})$ is an identity subelection of $E$ of size $(4,3)$. Also, $(\{a,b,c\}, V)$ is an antagonism subelection of $E$ of size $(3, 6)$; and $\{a,b,c\}$ is a clone set for the subelection $(C, \{v_1, v_2, v_4, v_5, v_6\})$ of $E$ of size $(3, 5)$.  





Given two identity (resp. antagonism) subelections $E_1, E_2$ for election $E$, we say that $E_1$ is larger than $E_2$ if $size(E_1) = (m_1, n_1)$, $size(E_2) = (m_2, n_2)$, $m_1 \geq m_2$, and $n_1 \geq n_2$, with one of these two inequalities being strict. In other terms, this is Pareto-dominance for two criteria, being the number of candidates and the number of voters in subelections. Note also that the existence of an identity of size $m'$ for $n'$ voters implies the existence of an identity of size $m'' \leq m'$ for $n'' \leq n'$ voters.  

We define the {\em identity} (resp. {\em antagonism}) signature of an election as the set of pairs $(m',n')$ of integers such that (i) there is an identity (resp. antagonism) subelection $E'$ of $E$ of size $(m',n')$, and (ii) no identity (resp. antagonism) subelection of $E$ is larger than $E'$.  
In Example \ref{running}, the identity signature of $E$ is $\{(1,6), (3,4), (4,3), (6,1)\}$. We will further say that a set of voters (resp. candidates) has an identity (resp. antagonism) subelection if there is a subelection of that type with that set of voters or candidates.

We assume that the reader is familiar with basic concepts in (parametrized) computational complexity. 

\section{Computing Meaningful Subelections}\label{sec:complexity}


In this section, we study the complexity of discovering hidden subelections and provide algorithms for finding them. 

\subsection{Hidden Clones}\label{sec:IHC}

Let us commence with the problem of finding clone sets in an election. In $\hiddenclones$ we are concerned with checking the existence of sufficiently large sets of voters similar with respect to a set of candidates of a given size.

\begin{quote}
		\noindent $\hiddenclones$:\\
		\hspace*{-1em} \indent\textit{Input:} Election $E=(C,V)$, $m',n' \in \mathbb{N}$.\\
		\hspace*{-1em}\textit{Question:} Is there a 
        set of $n'$ voters $V' \subseteq V$ and a set of~$m'$ candidates $C' \subseteq C$ such that 
        $C'$ is a clone set for~$V'$?
	\end{quote}

\begin{theorem} \label{thm:hidden-clones-p}
     $\hiddenclones$ is P-time solvable.
\end{theorem}
\begin{proof}
    Let $(E, m', n')$ be an instance of $\hiddenclones$. A subset of $m'$ candidates can be a clone set for at least one voter if and only if it forms a segment of $m'$ consecutive candidates in some vote. 
    We observe that there are at most $|V| \cdot (|C| - m' + 1)$ such subsets, we iterate over them and accept if any of them appear in~at~least $n'$ votes. 
\end{proof}



To better understand the algorithm, let us analyze it in the previous example.

\begin{example}
Consider again the election from Example \ref{running}. Take $m' = 3$: The segments of length 3 for voter 1 are $\{a,b,c\}$, $\{b,c,f\}$, $\{c,e,f\}$ and $\{d,e,f\}$; for voter 2, these are  $\{a,b,c\}$, $\{a,b,d\}$, $\{a,d,e\}$ and $\{d,e,f\}$; and so on. Counting the number of voters for which these subsets correspond to a segment of length 3, we find that $\{a,b,c\}$ appears four times, $\{b,c,f\}$ only once, $\{c,e,f\}$ twice, and so on. The sets occurring most frequently are $\{a,b,c\}$, $\{d,e,f\}$ (four times each): In conclusion, there are hidden clones of size $(3,n')$, for $n' \leq 3$.
\end{example}

The following observation states that the existence of clone sets of a given size is not monotonic (whereas for identity and antagonism subelections, monotonicity holds.)

\begin{observation}
The existence of a clone set of size $m'$ does not imply the existence of a clone set of smaller (nor larger) size.
\end{observation}

As a trivial example, $(a \succ b \succ c, b \succ c \succ a, c \succ a \succ b)$ has no clone set of two candidates, but every singleton, as well as the set of all candidates, are clone sets for all voters.

For this reason, we do not define the clone signature of an election, but for every $m' \leq m$ we define $\maxclone(E,m')$ as
the largest $n' \in [n]$ such that there exists a set of $n'$ voters $V' \subseteq V$ and a set of $m'$ candidates $C' \subseteq C$ such that $C'$ is a clone set. 

\begin{example}\label{ex:maxClone}
Let us consider the instance defined in~\Cref{running}. There, we observe that $\maxclone(E,2) = 6, \maxclone(E,3) = 5, \maxclone(E, 4) = \maxclone(E, 5) = 3$, and trivially, \\ $\maxclone(E,1) = \maxclone(E,6) = 6$. 
\end{example}

Analogously to the algorithm described in Theorem \ref{thm:hidden-clones-p}, we can also compute $\maxclone$ in polynomial time (we search for a subset with maximum occurrences).

\begin{corollary}\label{corr:max-clones-p}
    $\maxclone$ is P-time solvable.
\end{corollary}

We note that using standard data structures such as hash map and doubly linked list, both of our algorithms can be implemented in $O(|V| \cdot (|C| - m') \cdot m')$ time and space complexity, which makes them very fast and usable in practice. For example, for elections with a few hundred voters and a few hundred candidates, the algorithm works in a few seconds.

 Due to the nonmonotonicity of clone sets, it is meaningless to approximate the size of a maximal clone set. Regarding the approximation of the maximal number of voters for which there exists a clone set of a given size, solving such a problem is very similar to solving the $\maxclone$ problem.

 

We may also be interested in how far we are from obtaining a clone set of a given size for a given number of voters in terms of swap distance. 
However, the closest clone set (in terms of swap distance) may not be present in any vote: for $C = \{a,b,c,d,e,f,g,h,i \},$ 
$v_1 = \{a \succ b \succ c \succ d \succ e \succ f \succ g \succ h \succ i\}, 
v_2 = \{d \succ f \succ g \succ i \succ a \succ b \succ c \succ h \succ e\}, 
v_3 = \{a \succ b \succ c \succ f \succ e \succ i \succ d \succ h \succ g\}$, 
$n'=3$, $m'=4$, the closest clone set is $\{a,b,c,e\}$ with swap distance $3$. 
However,  for any clone set candidate $cs$,  we can compute in polynomial time the minimum number of swaps we need to do so that $cs$ becomes a valid clone set (and use it to find the best one).

\subsection{Hidden Identity}\label{sec:FHI}

We study the problem of whether a certain number of voters agree regarding the order of a given number of candidates. 
\begin{quote}
		\noindent $\hiddenidentity$:\\
		\hspace*{-1em} \indent\textit{Input:} Election $E=(C,V)$, $m',n' \in \mathbb{N}$.\\
		\hspace*{-1em}\textit{Question:} Is there an identity 
  subelection $E'=(C',V')$ of $E$, with $|C'| \geq m'$ and~$|V'| \geq n'$?
	\end{quote}





 We first show that $\hiddenidentity$ is intractable. 
 

\begin{theorem}\label{thm:ihidden-dentity-np}
$\hiddenidentity$ is \np-complete.
\end{theorem}

\begin{proof}[Proof Sketch]
 Membership to $\np$ is clear as we can guess both a set of voters and a set of candidates and check if these voters rank these candidates in the same order. Hardness is shown by reduction from $3$-SAT. 
 For a $3$-CNF formula $\phi$ with the set of variables $X=\{x_0, \dots, x_n \}$ and the set of clauses $C= \{C_0, \dots, C_m\}$, we consider a sufficiently large number of pairs of voters $M$, whom we call \emph{main voters}. We further include three \emph{clause voters} for each of the clauses in $C$. We will associate every such voter with a distinct literal in a corresponding clause. Subsequently, we consider a set of variables $X'$, consisting of $X$ and five additional variables for each clause, assuming an order of candidates in $X'$. Then, we construct a pair of \emph{literal candidates} for each such variable $x_i$, i.e., $c_{x_i}$ and $c_{\neg x_i}$. Furthermore, in the constructed instance of $\hiddenidentity$ we take the identity with the number of candidates equal to $|X'|$, and the number of voters $|C| + M$.

Further, we let main voters rank the candidates following the order of $X'$, with one of such voters in each pair ranking $c_{x_i} \succ c_{\neg x_i}$ and the other $c_{\neg x_i} \succ c_{ x_i}$, for every $x_i \in X'$. Notice that by choosing a sufficiently high number of main voters we ensure that in a subelection required in the instance we consider, we select exactly one literal candidate for each variable. Hence, such a subelection corresponds to some valuation over $X'$. We let every clause voter rank the candidate encoding the negation of a literal that the voter corresponds to lower than the candidates representing a variable with a higher index. Observe that then, if a clause voter $v$ corresponding to a literal $L$ is selected in a target subelection, then $c_{\neg L}$ is not, as then $v$'s vote would not be identical to the main voters' rankings. Subsequently, using additional variables in $X'$ corresponding to a clause $C_i$, we ensure that at most one of the clause voters for $C_i$ is present in a subelection satisfying criteria of our instance, because otherwise the selected clause voters would not have identical votes in the target subelections. As by the size of a target subelection we need to select at least $|C|$ clause voters, we obtain that it exists exactly when $\varphi$ is satisfiable. (See the full version of the paper for an example of votes in an encoding we define).
\end{proof}

Following the $\np$-hardness of the problem we consider in general, a natural approach is to ask for $\fpt$ algorithms. We will show that we are able to efficiently verify if a given set of voters has an identity of a given size (even if we do not know which candidates should be selected). This enables us to provide an $\fpt$ algorithm parameterized by the number of voters.  

\begin{proposition} \label{thm:hidden-identity-voters-verification-p}
    Checking if for a given set of $n'$ voters there exists
    an identity subelection with at least $m'$ candidates is P-time~solvable. 
\end{proposition}
\begin{proof}
    Suppose that we are given a $\hiddenidentity$ instance $(E, m', n')$ and a set of $n'$ voters $V' = \{v_{i_1}, v_{i_2}, \dots, v_{i_{n'}}\}$. We ask if there exist $m'$ candidates that are preferred exactly in the same order by all voters from $V'$.\footnote{Initially this approach is reminiscent of the (NP-hard) Longest Common Subsequence problem. However, 
    in our problem, each candidate appears in each vote exactly once. } 
    We define the {\em unanimity graph} $\textit{Una}(V')$ as the graph whose set of vertices is $C$, and that contains edge $(c,c')$ if and only if all voters in $V'$ prefer $c$ to $c'$. $\textit{Una}(V')$ is a directed acyclic graph, and $(C', V')$ is an identity subelection of $E$ if and only if there is a path in $\textit{Una}(V')$ that goes through all candidates of $C'$. Therefore, it suffices to check whether  $\textit{Una}(V')$ contains a path of length $m'-1$; as 
    finding the longest path in DAG is well-known to be P-time solvable \cite{cor-lei-riv-ste:b:algorithms-fourth-edition}, our algorithm runs in polynomial time. 
\end{proof}
Specifically, the Algorithm \ref{thm:hidden-identity-voters-verification-p} can be implemented in $O(n' \cdot |C|^2)$ time complexity and $O(n' \cdot |C| + |C|^2)$ space complexity. 
We demonstrate our approach in an example.

\begin{example}
We continue Example \ref{running}. Take  $V' = \{v_1,v_3,v_6\}$: The unanimity graph for $V'$ consists of the edges $a \rightarrow b$, $a \rightarrow c$, $a \rightarrow d$, $a \rightarrow e$, $b \rightarrow c$, $b \rightarrow d$, $c \rightarrow d$, $e \rightarrow d$, $f \rightarrow d$, $f \rightarrow e$. The longest path being $a \rightarrow b \rightarrow c \rightarrow  d$, $V'$ has a subelection with $4$ candidates (and no more).  With $V'' = \{v_1,v_2,v_3,v_6\}$, the unanimity graph is composed of the edges $a \rightarrow d$, $b \rightarrow c$, $b \rightarrow d$, the longest path is $b \rightarrow c \rightarrow  d$, $V''$ has a subelection with $3$ candidates (and no more).
\end{example}

Using the above algorithm, we obtain{that} $\hiddenidentity$ is fixed-parameter tractable for the number of voters.

\begin{corollary} \label{thm:hidden-identity-ftp-n}
    $\hiddenidentity$ is in $\fpt$ for the parameterization by the number of voters ($|V|$) and in $\xp$ 
    for the parameterization 
    by the number of voters in the solution ($n'$).
\end{corollary}
\begin{proof}
    Suppose we are given a $\hiddenidentity$ instance $(E, m', n')$. We iterate through all possible size-$n'$ subsets of voters and check if the algorithm \ref{thm:hidden-identity-voters-verification-p} found any identity subelection consisting of at least $m'$ candidates, if so, then we accept, otherwise we reject.
    The algorithm works in ${O({|V| \choose n'} \cdot (n' \cdot (n' \cdot |C|^2)))}$ 
    time complexity and $O(|V| \cdot |C| + |C|^2)$ space complexity.
\end{proof}

We show analogous results for a given set 
of candidates and for the parameter number of candidates.

\begin{proposition} \label{thm:hidden-identity-candidates-verification-p}
    Checking if for a given set 
    of $m'$ candidates there exists an identity subelection with at least $n'$ voters is P-time solvable.
\end{proposition}
\begin{proof}
    Take a $\hiddenidentity$ instance $(E, m', n')$ and a set of $m'$ candidates $C' = \{c_{i_1}, c_{i_2}, \dots, c_{i_{m'}}\}$. We ask if there are $n'$ voters that rank them identically. 
    As there are at most $|V|$ distinct orders of candidates $C'$ in our instance, it suffices to check if any of them appears in at least $n'$ votes.
\end{proof}

The Algorithm \ref{thm:hidden-identity-candidates-verification-p} can be implemented in $O(|V| \cdot |C|)$ time and space complexity.
Due to it, we obtain parameterized tractability of $\hiddenidentity$ for the number of candidates.

\begin{corollary} \label{thm:hidden-identity-fpt-m}
    $\hiddenidentity$ is~in~$\fpt$ for the parameterization by the number of candidates~($|C|$) and in~$\xp$ 
    for the parameterization by the number of candidates in the solution~($m'$).
\end{corollary}
\begin{proof}
    Suppose we are given a $\hiddenidentity$ instance $(E, m', n')$. We iterate through all size-$m'$ sets of candidates and accept if the algorithm \ref{thm:hidden-identity-candidates-verification-p} found any identity subelection consisting of at least $n'$ voters, otherwise we reject. 
    The algorithm works in $O({|C| \choose m'} \cdot (|V| \cdot |C|))$ time complexity and $O(|V| \cdot |C|)$ space complexity.
\end{proof}

\begin{figure*}
  \begin{subfigure}[b]{0.3333\textwidth}
      \centering
      \includegraphics[width=6.1cm, trim={0.1cm 0.1cm 0.1cm 0.1cm}, clip]{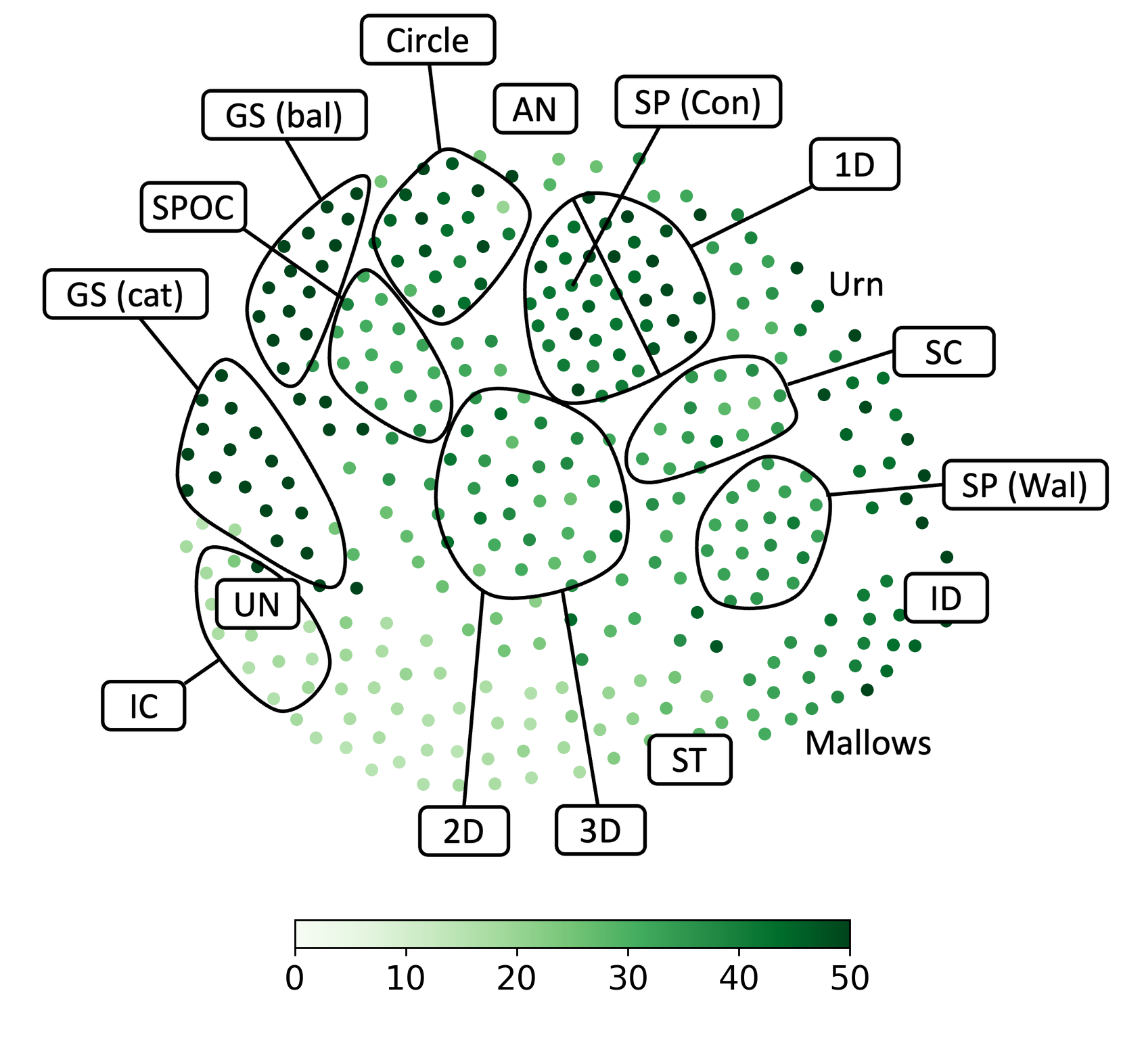}
      \caption{ $\maxclone(E,2)$.}
  \end{subfigure}%
    \begin{subfigure}[b]{0.3333\textwidth}
      \centering
      \includegraphics[width=6.1cm, trim={0.1cm 0.1cm 0.1cm 0.1cm}, clip]{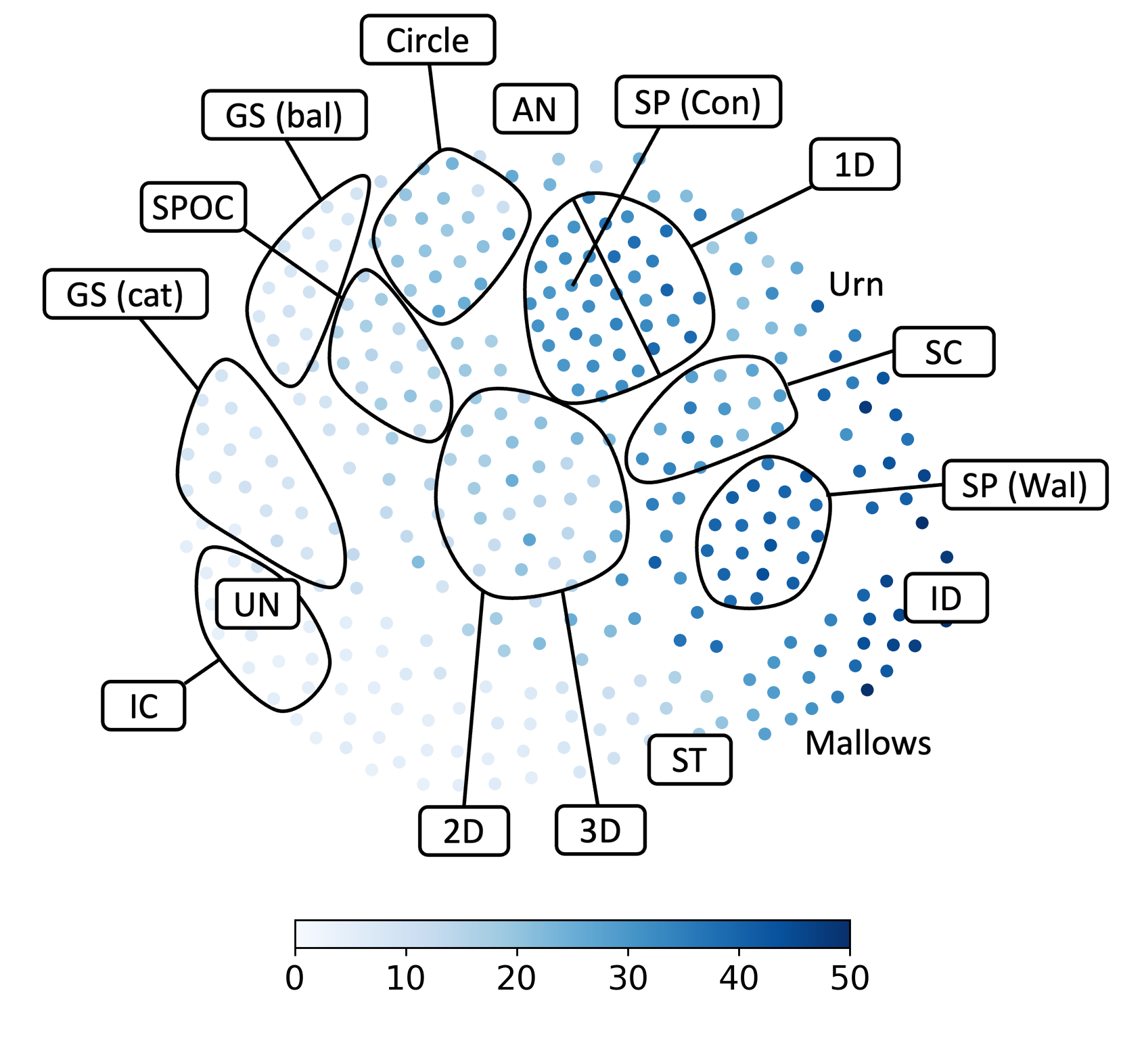}
      \caption{ $\maxid(E,5)$}
  \end{subfigure}%
    \begin{subfigure}[b]{0.3333\textwidth}
      \centering
      \includegraphics[width=6.1cm, trim={0.1cm 0.1cm 0.1cm 0.1cm}, clip]{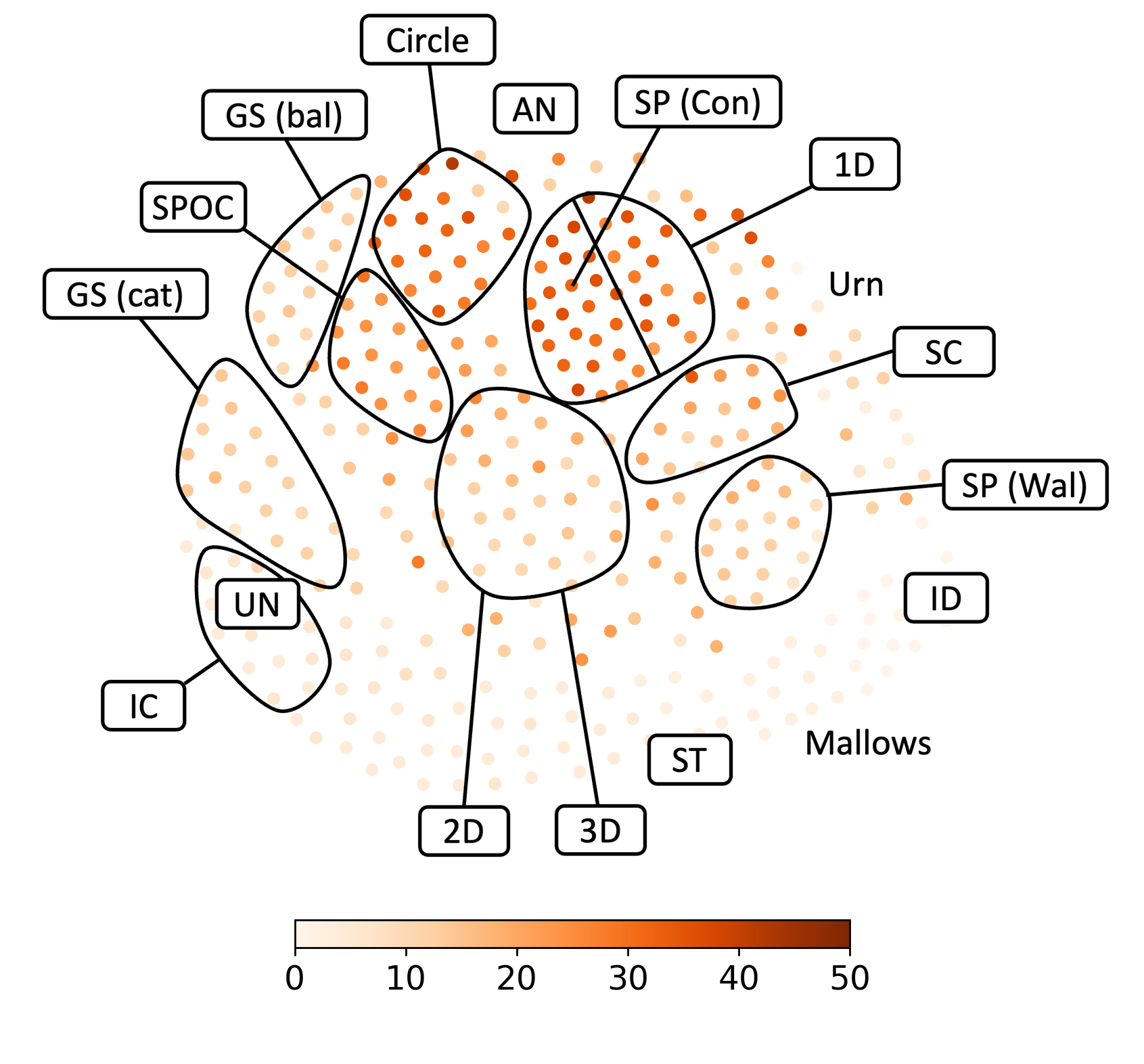}
      \caption{ $\maxan (E,5)$}
  \end{subfigure}%

   \caption{\label{exp:maps} 
   Maps of elections with 10 candidates and 50 voters. Each point represents a single election, and its color represents the maximum number of voters that a) find certain two candidates clones (left), b) agree on certain five candidates being identity (middle), c) are antagonized over certain five candidates. In other words, the darker the point is, the more voters agree on a certain set of candidates being clones (left), identity (middle), or antagonism (right). On each map, ID label marks the identity election, and AN label marks the antagonism election, and dots representing elections coming from the same statistical culture were connected in clusters with names.}
\end{figure*}

  The algorithms described in Propositions \ref{thm:hidden-identity-voters-verification-p} and \ref{thm:hidden-identity-candidates-verification-p} can be used to effectively answer natural questions about preferences such as what do voters coming from a certain background mostly agree on, or are certain alternatives ranked in the same order by some large group of voters. 
  As we showed in Corollaries \ref{thm:hidden-identity-ftp-n} and \ref{thm:hidden-identity-fpt-m}, we are able to answer these questions effectively provided that either the number of voters or the number of candidates is relatively small (up to $20$ or $30$).
  Nevertheless, if both the number of voters and the number of candidates are not small, then this approach is too slow.
  To tackle this problem, we provide an ILP for $\hiddenidentity$ which finds an identity subelection (if exists) or the closest subelection to identity if such does not exist. 


\begin{proposition}\label{ILP:identity}
  There is an ILP for \hiddenidentity{} which selects a solution for a ``yes''-instance and the closest subelection to identity (in terms of swap distance) for a ``no''-instance.
\end{proposition}

\begin{proof}

  Let~$E = (C,V)$ be the election we wish to
  analyze, with $C = \{c_1, \dots, c_m\}$ and~$V = \{v_1, \dots, v_n\}$. All variables will be binaries. For each~$i \in [n]$, we define 
  a variable~$V_{i}$ with the intention that value~$1$ indicates that
  voter~$v_i$ is selected. Similarly, for each~$j \in [m]$, we define 
  a binary variable~$C_{j}$ with the intention that value~$1$ indicates that
  candidate~$c_j$ is selected. The variable $S_{j_1,j_2}$ is equal to $1$ if candidates $C_{j_1}$ and $C_{j_2}$ are selected and candidate $C_{j_1}$ appears before $C_{j_2}$ in the identity ranking.
  Variable $P_{i,j_1,j_2}$ is equal to $1$ if voter $V_i$ agrees that $C_{j_1}$ is ranked before  $C_{j_2}$ and both these candidates are selected.
  We introduce the following constraints: \footnote{
 In some constraints, we used the multiplication operation. However, as all the variables are binaries they can be easily replaced by standard constraints of a less intuitive form.}
  \begin{align}
    \label{ilp:i1}
    &\textstyle\sum_{i \in [n]} V_{i} = n', \\
    \label{ilp:i2}
    &\textstyle\sum_{j \in [m]} C_{j} = m', \\
    \label{ilp:i3}
    &S_{j_1,j_2} + S_{j_2,j_1} = C_{j_1} \cdot C_{j_2}, \text{\quad}\forall_{j_1,j_2 \in [m]}, \\
    \label{ilp:i4}
    &P_{i,j_1,j_2} = V_i\cdot S_{j_1,j_2}, \text{\quad}\forall_{i \in [n], j_1,j_2 \in [m]}.
  \end{align}
  Constraints~\eqref{ilp:i1} and~\eqref{ilp:i2} ensure that we select the proper numbers of voters and candidates.
   Constraints~\eqref{ilp:i3} and~\eqref{ilp:i4} implements the logic of $S$ and $P$ variables, respectively.
  The optimization 
  goal is to minimize:
 $\textstyle\sum_{i \in[n],\ j_1, j_2 \in [m]} P_{i,j_1,j_2} \cdot
  W_{i,j_1,j_2}$,
where $W_{i,j_1,j_2} = [pos_{v_i}(c_{j_1}) > pos_{v_i}(c_{j_2})]$.
\end{proof}




We see that with the proposed ILP, we are able to maximize the solution if exists as well as to effectively determine the closest subelection if the exact one does not exist.
We also consider the maximization problem for $\hiddenidentity$, which will be useful for our experiments.
For $m' \in [m]$, by $\maxid(E,m')$ we denote the problem of finding $n'\in[n]$ such that ($n',m'$) is the identity signature, i.e., $size(E) = (m', n')$. We can solve it via a modification of the ILP in Proposition \ref{ILP:identity}, that is, by 1) adding a constraint $\sum_{i \in [n]} V_{i}$ (i.e., the previous objective function) and 2) changing the optimization goal to maximize $\sum_{i \in [n]} V_{i}$.

\subsection{Hidden Antagonism}\label{sec:IHA}

We further analyze the problem of checking if an election contains an antagonism of a given size.

\begin{quote}
		\noindent $\hiddenantagonism$:\\
		\hspace*{-1em} \indent\textit{Input:} Election $E=(C,V)$, $m',n' \in \mathbb{N}$.\\
		\hspace*{-1em}\textit{Question:} Is there an antagonism subelection $E'=(C',V')$ of $E$, with $|C'| \geq m'|$,~$|V'| \geq n'$? 
	\end{quote}

First, we show that $\hiddenantagonism$ is intractable in general.

\begin{theorem}\label{thm:IHANP}
$\hiddenantagonism$ is \np-complete.
\end{theorem}
\begin{proof}[Proof Sketch]
The proof extends the reduction in the proof Theorem \ref{thm:ihidden-dentity-np}. For a 3-CNF formula $\varphi$ with the set of variables $X=\{x_0, \dots, x_n\}$ and the set of clauses $C= \{C_0, \dots, C_m\}$, 
we begin by taking the encoding provided there. Subsequently, we double the number of main voters, reversing their order of candidates corresponding to particular variables for half of these voters. So, for half of the main voters, we have that $c_{x_i}$ and $c_{\neg x_i}$ are preferred to both $c_{x_j}$ and $c_{\neg x_j}$, while for a half $c_{x_j}$ and $c_{\neg x_j}$ are preferred to $c_{x_i}$ and $c_{\neg x_i}$, if $i >j$. 
Observe that as in the case of the reduction in the proof of Theorem \ref{thm:ihidden-dentity-np}, by choosing a sufficient number of main voters and the required number of candidates in a target subelection as $|X'|$, we ensure that it corresponds to some valuation over $X'$. 

Then, we take two copies of voters corresponding to each clause, requiring voters in each such copy 
to have a reversed order of candidates corresponding to particular variables, as in the case of main voters. Following the reasoning provided in the proof of Theorem \ref{thm:ihidden-dentity-np}, due to the reversed order of candidates,  there exists an antagonism with at least $|X'|$ candidates and $2|C| + M$ voters, where $M$ is the number of pairs of main voters, exactly when $\varphi$ is satisfiable. 
\end{proof}

Although it requires more effort to model antagonized votes than identical ones, we can still compute maximum antagonism for a given set of voters or candidates in polynomial time and thus obtain $\fpt$ algorithms for \hiddenantagonism{}.



\begin{proposition} \label{thm:hidden-antagonism-voters-verification-p}
    Checking if for a given set of $n'$ voters there exists an antagonism subelection 
    with at least $m'$ candidates is P-time~solvable.
\end{proposition}
\begin{proof}
    Suppose we are given a $\hiddenantagonism$ instance $(E, m', n')$, as well as a set of $n'$ voters $V' = \{v_{i_1}, v_{i_2}, \dots, v_{i_{n'}}\}$. We ask if there exists a set of $m'$ candidates such that all voters from $V'$ are antagonized over them. That is, half of them rank these candidates in the same order, whereas the other half in the opposite one.

    For $m' = 1$ the answer is naturally yes, so we assume that $m' \geq 2$. 
    Suppose now that the answer is yes, the solution is order $p'$ of candidates $C' \subseteq C$, candidate $c_b$ is at the beginning of $p'$, and candidate $c_e$ is at the end of $p'$.
    Let $V'_{c_b \succ c_e}$ and $V'_{c_e \succ c_b}$ be the voters from $V'$ preferring $c_b$ over $c_e$ and $c_e$ over $c_b$, respectively.
    Then all voters from $V'_{c_b \succ c_e}$ must order candidates $C'$ in the order $p'$ and all voters from $V'_{c_e \succ c_b}$ must order them exactly in the reversed order $r'$.
    However, all voters from $V'_{c_e \succ c_b}$ order candidates $C'$ exactly in the reversed order $r'$ if and only if they order them in the order $p'$ after reversing their whole votes.

    With this observation, we propose a polynomial-time algorithm as follows.
    We iterate over each pair of distinct candidates $c_b, c_e$ from $C$ (with the intention that they will be, respectively, the first and the last candidate in the antagonism subelection).
    If $|V'_{c_b \succ c_e}| \neq |V'_{c_e \succ c_b}|$, then the sets of antagonized voters do not have the same cardinality, so we continue with the next pair.
   Now we know that $|V'_{c_b \succ c_e}| = |V'_{c_e \succ c_b}|$. We remove all candidates from the votes that appear in at least one vote not between $c_b$ and $c_e$ (i.e., we keep candidate $c$ if and only if for each vote either $c_b \succ c \succ c_e$ or $c_e \succ c \succ c_b$).
    We are now left with the truncated votes consisting of $c_b$ and $c_e$ ranked at the extreme positions and all remaining candidates ranked between them in all votes.
    If we end up with less than $m'$ candidates, then we continue with the next pair.
    We reverse votes in which $c_e \succ c_b$, use the algorithm \ref{thm:hidden-identity-voters-verification-p} to see if there is an identity subelection with $m'$ candidates, and accept if there is one (note that due to $|V'_{c_b \succ c_e}| = |V'_{c_e \succ c_b}|$ it is equivalent to having half of voters approving one order and the other half preferring the opposite one).
    We reject if we do not find any solution for any pair of candidates $c_b, c_e \in C$.
    The algorithm runs in $O(n' \cdot |C| + |C|^2 \cdot (n' \cdot |C|^2))$ time and uses $O(n' \cdot |C| + |C|^2)$ space.
\end{proof}


\begin{proposition} \label{thm:hidden-antagonism-candidates-verification-p}
    Checking if for a given set 
    of $m'$ candidates there exists an antagonism subelection with at least $n'$ voters is P-time~solvable.
\end{proposition}
\begin{proof}
    Suppose we are given a $\hiddenantagonism$ instance $(E, m', n')$ and a set of $m'$ candidates $C' = \{c_{i_1}, c_{i_2}, \dots, c_{i_{m'}}\}$. We ask if there exists a set of $n'$ voters that are antagonized over them, that is, half of them rank candidates $C'$ in one way and the other half in the opposite order.
    We create a hash map $D$ mapping orders (permutations) of candidates from $C'$ to lists of voters that rank them in this order. For the sake of brevity, let $D[p]$ be the value in $D$ associated with the key $p$, i.e., the list of voters that rank candidates $C'$ in the order $p$.
    Then, we iterate over each order $p$, in the votes of $D$ and accept if for any order $p$, both $D[p]$ and $D[r]$ contain at least $\nicefrac{n'}{2}$ voters where $r$ is the reversed order of $p$. 
    Analogously to the algorithm in Proposition \ref{thm:hidden-identity-candidates-verification-p}, the algorithm runs in polynomial time because we consider only permutations that appear in the given votes. 
    Specifically, its time and space complexity is $O(|V|\cdot|C|)$.
\end{proof}

Then, fixed-parameter tractability of $\hiddenantagonism$ for the number of voters and the number of candidates follows.

\begin{corollary} \label{thm:hidden-antagonism-ftp-n-and-fpt-m}
    $\hiddenantagonism$ is~in~$\fpt$ parametrized by the number of voters~($|V|$) or 
    candidates~($|C|$) as~well~as in~$\xp$ 
    for the parameterization
    by the number of voters  ($n'$) or candidates ($m'$) in the solution.
\end{corollary}

Analogously to \hiddenidentity{}, these FPT algorithms suffice if either the number of voters or the number of candidates is small, but they are too small if both of these values are not small. To handle it, we use a simple ILP for that problem. 
(The details are in the full version of the paper).


\begin{proposition}\label{prop:ilphiddenant}
  There is an ILP for \hiddenantagonism{} which selects a
solution for a “yes”-instance and the closest subelection to antagonism (in
terms of swap distance) for a “no”-instance.
\end{proposition}

The strong advantage of this ILP is that it also manages situations in which the desired antagonism does not exist, which, as we will see in experiments, is not a rare case.

Additionally, we study a maximization version of $\hiddenantagonism$ that will be crucial regarding our experiments.
Given $m' \in [m]$, by $\maxan(E,m')$ we denote the problem of finding $n'\in[n]$  such that ($n',m'$) is the antagonism signature, i.e., $size(E) = (m', n')$.

 %
%
%
%
We note that the ILPs provided in this section, i.e., in Propositions \ref{ILP:identity} and \ref{prop:ilphiddenant}, as well as those for \maxid{} and \maxan{}, are crucial for the experiments we provide in the next section. 

\section{Experiments}\label{sec:EXP}
In this section, we focus on the practical application of our approach. First, we study our problems with synthetic data. Later on, we analyze several real-life instances.

\subsection{Map of Elections}

To depict our experimental results, we use the framework introduced by \citet{szu-fal-sko-sli-tal:c:map} and extended by \citet{boe-bre-fal-nie-szu:c:compass}, known as \emph{map of elections}. 
The map serves us to better understand the space of elections and is particularly useful when conducting experiments. Each point on the map depicts a single election. The embeddings were calculated based on the mutual distances between elections (computed with some distance function). The closer the two points on the map are, the more similar the elections these points depict. For instance, elections coming from the same statistical cultures or similar models are often clones on the map, while elections coming from very different models are usually more distant on the map.
In this case, we use the map from \citet{boe-fal-nie-szu-was:c:metrics} that consists of $344$ elections\footnote{Detailed description is provided in the full version of the paper.} with $10$ candidates and $50$ voters. The map is based on the isomorphic swap distance~\cite{fal-sko-sli-szu-tal:c:isomorphism}, and is embedded using Fruchterman-Reingold  algorithm~\cite{fruchterman1991graph}. 

Having prepared the map and its points, one can put on them some properties hidden in the colors or shapes of points. E.g., \citet{szu-fal-jan-lac-sli-sor-tal:c:map-approval} color the points on the map of approval elections with several statistics (e.g., cohesiveness level, PAV run time, and maximum approval score) to see how well statistical cultures and elections situated in different regions satisfy these properties or how well they perform. Here we will conduct a similar analysis of values distribution as well as its location on the map. We also investigate why this particular coloring occurred and what it means.

For each election from a given dataset, we computed the following characteristics: 
(1)~$\maxid(E,5)$,
(2)~$\maxan(E,5)$,
(3)~$\maxclone(E,2)$. 
The results are presented in~\Cref{exp:maps}. (Complementary results, containing separate average values for each statistical culture, are presented in the full version of the paper).

In the case of $\maxclone(E,2)$ (as depicted in the leftmost map) the darker the point, the more voters agree that there exists a pair of clone candidates. The most modest values are consistently seen in elections originating from the impartial culture (i.e., each vote is sampled uniformly at random), averaging at a value of $17$, with a remarkably low standard deviation of only $1.22$. 
Note that the value for {$\maxclone(E,2)$} seems not to be strongly correlated with the position on the map. It is because the map is based on swap distance. By making relatively few swaps (i.e., increasing a distance just a bit), we can significantly decrease the number of voters agreeing that two candidates are clones. In other words, given an election $E$ we can create a~new election $E'$ that is very close to $E$ (distance-wise) but has a much lower {\maxclone} value. In principle, for the {\maxclone} problem, usually by one swap we can lower the results by one (unless that swap is creating a~new solution involving a~different set of candidates). 

For the $\maxid(E,5)$ (the middle map) we observe a~strong correlation between the number of voters agreeing on given five candidates being ranked in a particular order and the swap distance from ID (the Pearson correlation coefficient (PCC) is 
$-0.791$). Similarly, the results for $\maxan(E,5)$ (the rightmost map) are strongly correlated with the distance from AN (PCC = $-0.845$). 
Both of these correlations are reasonable, as the larger the hidden identity (resp. antagonism), the fewer swaps we need to convert the election into $\ID$ (resp. $\AN$). This means that for $\maxid$ and $\maxid$ there is a strong correlation between the position on the map and the size of the signature. 
Unlike for clones, for identity and antagonism, it is harder to ``spoil'' the inner substructure, especially when $m'$  is relatively small with regard to $m$. For instance, for a vote $a \succ b \succ c \succ d \succ e \succ f$, no matter which three candidates are forming the solution, we can always ``spoil'' this vote with just one swap. And for example for identity if $a$, $c$, and $f$ form a solution, we need at least two swaps to ``spoil'' this vote.



\begin{remark}
The way we define {\maxan} might seem too rigid, as we require exactly the same number of ``base'' and ``reverse'' votes. However, we have also verified two other approaches. One, where we maximized the sum of \#base and \#reverse votes; it turns out, that usually the outcome is more similar to the one provided by {\maxid} than by {\maxan}. The second approach is to use the product of \#base and \#reverse votes. There, 
the result was usually very similar to the one provided by the ``rigid" approach (PCC = 0.929), yet the running-time was significantly longer. Therefore, for the sake of simplicity, we focus on the ``rigid" approach. 
\end{remark}



\vspace{0.2cm}
\subsection{Real-Life Instances}
We conducted experiments on real-life instances. In particular, we focus on two datasets, i.e., the \emph{Sushi} dataset, where $5000$ people expressed their preferences over $10$ sushi types~\citep{kam:c:sushi}; and the
\emph{Grenoble} dataset, containing data from a field experiment held in Grenoble in 2017, where people expressed their preferences over French presidential candidates~\citep{grenoble-online}. There, we use the same data preprocessing method as~\citet{fal-kac-sor-szu-was:c:microscope}.

We run {\maxclone}, {\maxid}, and {\maxan} for all possible numbers of hidden candidates.
The results are presented in~\Cref{exp:joint_plots}. 
For all experiments, we include the respective values from impartial culture (IC) instances, which we treat as a lower bound.\footnote{Technically, it is possible to create an instance with even smaller clones/identity/antagonism than impartial culture, but the difference will not be of great importance.}(Each value is an average derived from $10$ different IC elections).

We first focus on hidden clones. Every individual candidate is seen as a clone by all the voters.
However, it is worth noting that if the number of hidden candidates is equal to the number of all candidates, then again all the voters would agree that all the candidates are clones, hence the characteristic ``U'' shape in~\Cref{exp:joint_plots}.

For Sushi, $38.5\%$ of voters agree that Tamago (egg) and Kappa-maki (cucumber roll) are clones. This is especially interesting, as these are the only two vegetarian options in that dataset.
For Grenoble,  the strongest set of two clones is Nathalie Arthaud and Philippe Poutou (clones for $46.7\%$ of voters), which can be explained by the fact that they have similar far-left ideologies; and the strongest clone set of size 3 is composed of Jean-Luc M\'elenchon (left), Benoît Hamon (centre-left) and Emmanuel Macron (centre-right), who are clones for $26.6\%$ of the voters. This may initially look surprising 
but can be explained knowing that (a)~all three candidates are major, well-known candidates; (b)~on the left-right axis, they are arguably contiguous; (c)~there was (at least in the voters of the dataset) a clear dividing line between the candidates from left to center-right on the one hand, and right and far-right candidates on the other.

We now shift to hidden identity. In the case of Sushi, an overwhelming $88.3\%$ of the voters agree that Toro (fatty tuna) is better than Kappa-maki (cucumber roll). What is more intriguing is that $43.4\%$ of the voters agree on the following ranking:  Toro (fatty tuna) is preferred over Maguro (tuna), which is, in turn, preferred over Tekka-maki (tuna roll), and finally, Kappa-maki (cucumber roll). This order of preference is especially interesting as it mirrors the price hierarchy of these sushi types. When it comes to the political elections dataset, $89.2\%$ of the voters prefer Benoît Hamon (center-left) to Marine Le Pen (far-right).

As to the hidden antagonism, we only briefly discuss the results for the sushi dataset. The pair Ika (squid) and Tekka-maki (tuna roll) antagonized the whole society ($99.9\%$ of voters). Moreover, $41.2\%$ of the voters agree or strongly disagree (casting reverse order) with the following ranking: Uni (sea urchin) $\succ$ Kappa-maki (cucumber roll) $\succ$  Tamago (egg). It is intriguing that both the Sushi and Grenoble datasets show minimal signs of antagonism. The results are almost the same as for IC elections.

\begin{figure}[h]
\centering
  \begin{subfigure}[b]{0.4\textwidth}
      \centering
      \includegraphics[width=\textwidth]{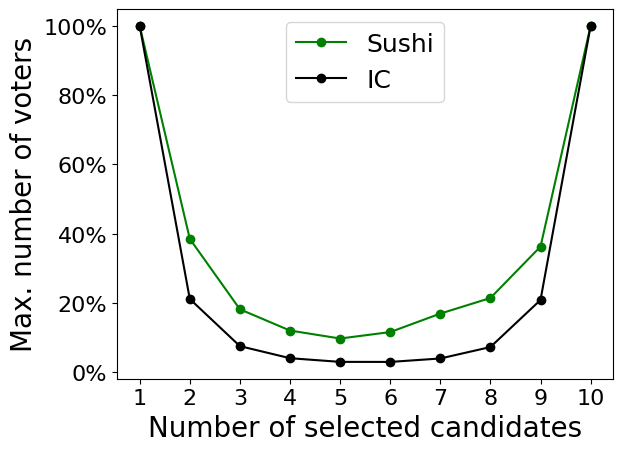}
      \caption{ \maxclone }
  \end{subfigure}\hspace{4em}%
    \begin{subfigure}[b]{0.4\textwidth}
      \centering
      \includegraphics[width=\textwidth]{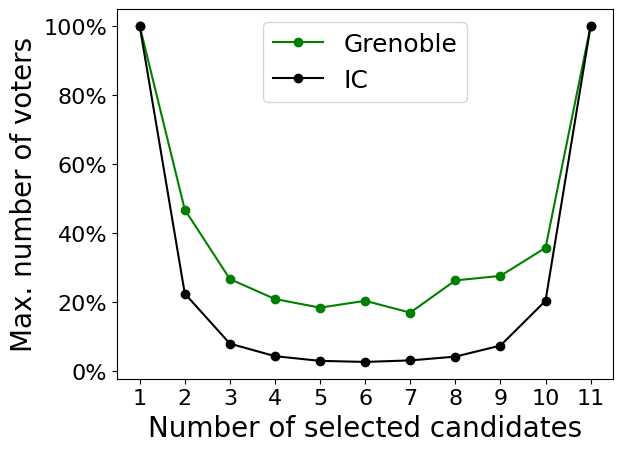}
      \caption{\maxclone }
  \end{subfigure}

  \vspace{0.3cm}
  
    \begin{subfigure}[b]{0.4\textwidth}
      \centering
      \includegraphics[width=\textwidth]{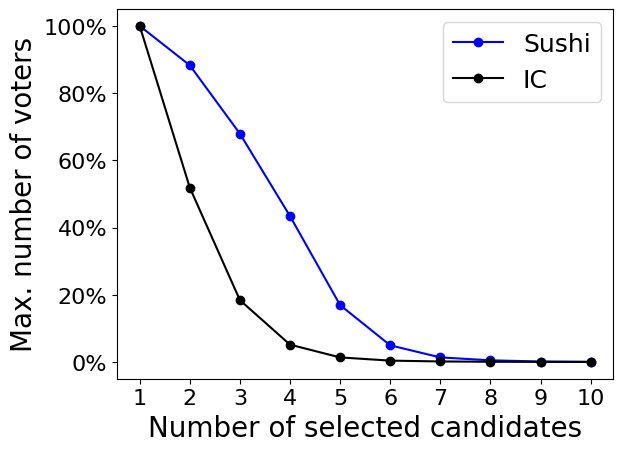}
      \caption{ \maxid }
  \end{subfigure}\hspace{4em}%
    \begin{subfigure}[b]{0.4\textwidth}
      \centering
      \includegraphics[width=\textwidth]{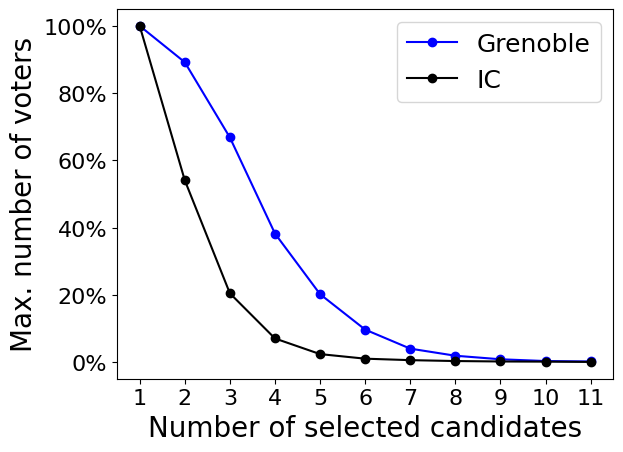}
      \caption{ \maxid }
  \end{subfigure}

  \vspace{0.3cm}
  
    \begin{subfigure}[b]{0.4\textwidth}
      \centering
      \includegraphics[width=\textwidth]{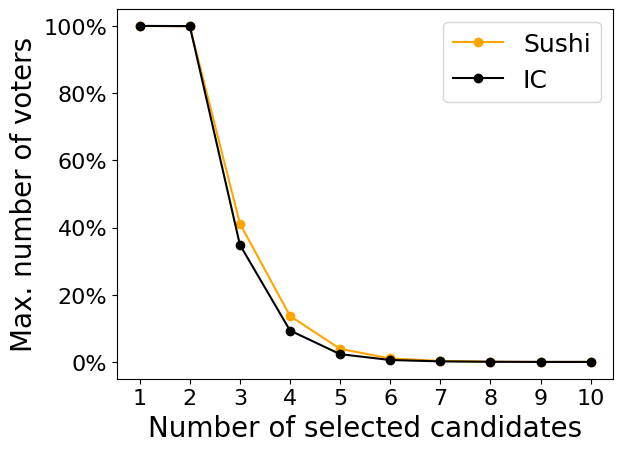}
      \caption{\maxan }
  \end{subfigure}\hspace{4em}%
    \begin{subfigure}[b]{0.4\textwidth}
      \centering
      \includegraphics[width=\textwidth]{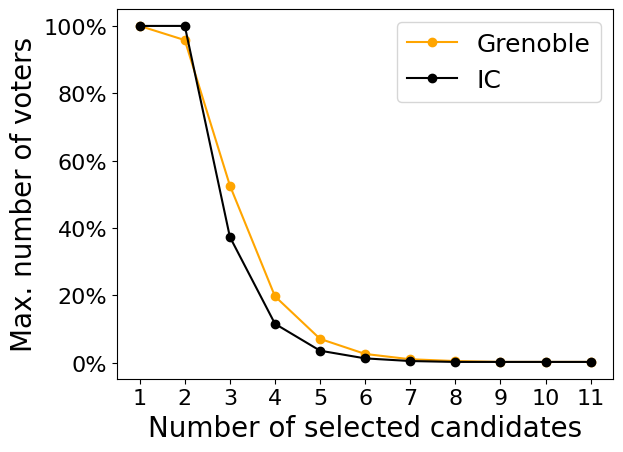}
      \caption{ \maxan }
  \end{subfigure}%

   \caption{\label{exp:joint_plots} Comparison of Sushi and Grenoble datasets. The black lines denote the results for impartial culture elections. }
   
\end{figure}


  

   

\section{Summary and Future Work}\label{sec:conclusion}
We explored the concept of hidden substructures in ordinal elections. We focused on three types of consistency: Identity, antagonism, and clones. 
We executed a comprehensive analysis of the complexity of the introduced problems and provided algorithms that can be used in practice. We showed as a possible direction the search for the closest subelection to the desired one if the exact one does not exist. Furthermore, we provided experimental evaluations on synthetic and real-life datasets. The experiments on real-life datasets confirmed that identifying consistent subelections indeed helps in learning interesting information hidden in an election.
Analyzing substructures of elections can help better understand different segments of the population and their preferences, which can be beneficial in a variety of contexts, such as consumer behavior or political opinion analysis.

 We see this paper as a starting point for a more thorough study. Indeed, there are numerous possibilities for further research related to the problems we considered.
To start with, we could soften our criteria and look for subelections that are near-identity or near-antagonism, and look for near-clones. 
 Next, the complexity of these issues could be analyzed when applied to structured domains, such as single-peaked or single-crossing domains. 
 A further natural extension could involve exploring approval elections.


 

		
			
   
   
	

\vspace{0.1cm}
\section*{Acknowledgements}
This project has received funding from the European Research Council
(ERC) under the European Union’s Horizon 2020 research and innovation
programme (grant agreement No 101002854), and from 
the French government under management of Agence Nationale de la 
Recherche as part of the ``Investissements d'avenir" program, reference ANR-19-P3IA-0001 (PRAIRIE 3IA Institute). 
We thank Florian Yger for fruitful discussions.
\begin{center}
  \includegraphics[width=3cm]{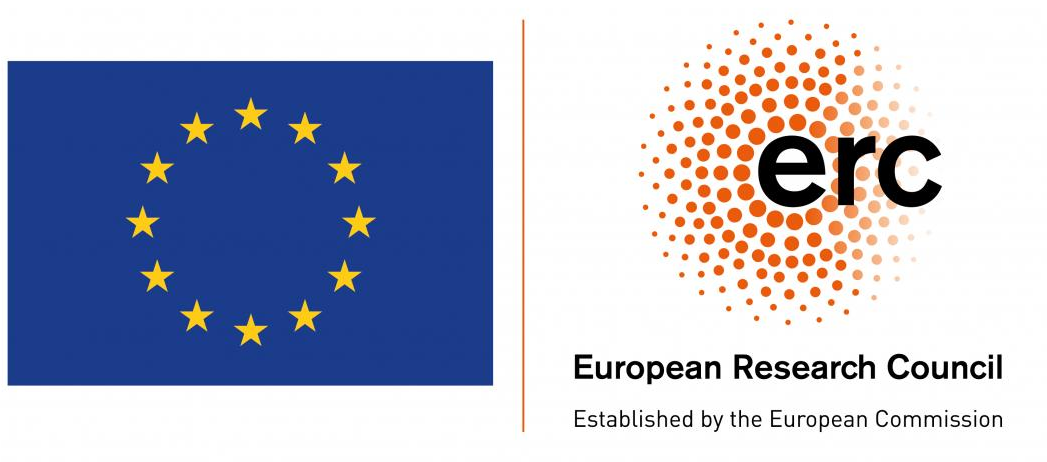}
\end{center}

\bibliographystyle{ACM-Reference-Format} 

\bibliography{discovering_consistent_subelections}


\newpage

\appendix

\section{Missing Computational Results for \textsc{Hidden-Clones}} \label{appendix:hidden-clones}

\paragraph{Minimizing Swap Distance for Hidden Clones}

\begin{proposition}
    For any clone set candidate $cs$, we can calculate in polynomial time the minimum number of swaps we need to do in some votes so that $cs$ becomes a valid clone set.
\end{proposition}
\begin{proof}
    Let $(E, m', n')$ be a $\hiddenclones$ instance and $sc$ be a given clone set candidate. As swaps in different votes do not affect other votes, it is enough to show how to compute the swap distance for a vote, then select $n'$ votes with the minimum swap distances, and finally return their sum. To calculate swap distance for a given vote, briefly saying, it suffices to see that in order to compute the distance of one vote from the closest vote containing $sc$ as a clone set, it is never beneficial to swap two candidates of $sc$ as well as to swap two candidates outside $sc$, because we do not care about the inner order. Thus we divide the vote into consecutive segments of non-``sc'' candidates alternating with ``sc'' candidates and compute the minimum number of swaps to make ``sc'' candidates appearing together via a simple algorithm (or dynamic programming). More precisely, we need to move each candidate not belonging to $sc$ either on the left of all candidates from $sc$, or on the right of all candidates from $sc$. It is easy to compute how many candidates from $sc$ I need to jump over to be on the left / on the right. We choose the minimum option for each candidate (please note that starting from the most extreme candidates to the most inner ones we will indeed have this number of swaps. 
    
    Finally, we use this algorithm for each vote, determine $n'$ votes with the minimum number of swaps, and return sum of their distances.
\end{proof}

\paragraph{Counting Hidden Clones}

We show that we can also count the number of pairs of size-$m'$ sets of candidates and size-$n'$ sets of voters that these candidates are clones for these voters.

\begin{corollary} \label{thm:sharp-hidden-clones-p}
  There is a P-time algorithm for $\sharphiddenclones$.
\end{corollary}
\begin{proof}
   Take a $\sharphiddenclones$ instance $(E, m', n')$
   Let $\textit{Cnt}$ be a dictionary transforming sets into the number of occurrences as consecutive segments in $E$. At the beginning, $\textit{Cnt}$ is empty. We iterate through all the voters and all the size-$m'$ consecutive segments in their voters. For each consecutive segment we construct a set of candidates that it contains (we denote it as $S$). If $S$ is already in $Cnt$, then we increment its value, otherwise we add key $S$ with value $1$. It is clear that it can be done in polynomial time.

   Having computed $\textit{Cnt}$, the answer is: $\sum_{S \in \textit{Cnt}} {\textit{Cnt}[S] \choose n'}$ where $\textit{Cnt}[S]$ means the number of occurrences of set $S$ in $\textit{Cnt}$. Binomial coefficients can be computed in polynomial time and stored in a separate lookup table. Please note that we only iterate through sets present is $\textit{Cnt}$, iterating through all possible size-$m'$ sets would result in exponential complexity.

    Using standard data structures such as hashmap and hashset, the time complexity of this algorithm is $O(|V| \cdot (|C|-m'+1) \cdot m' + |V| \cdot m') = O(|V| \cdot (|C|-m') \cdot m')$, space complexity is the same. 
\end{proof}

Due to this result, we can easily sample clones from the election or compute the number of clones for a given set of voters or candidates.

\section{Missing Hardness Proof of Theorem \ref{thm:ihidden-dentity-np}}

We will denote as $rank_v(c)$ the position of a candidate $c$ in $v$'s ranking. 

\begin{proof}
Membership to $\np$ is clear as we can guess both a subset of voters and a subset of candidates and check if these voters rank these candidates in exactly the same order.

We will show that $\hiddenidentity$ is \np-hard by reduction from $\threesat$. Consider a $3$-CNF formula $\varphi$ with the set of variables $X=\{x_0, \dots, x_n \}$ and the set of clauses $C= \{C_0, \dots, C_m\}$. We assume that each clause consists of exactly three literals. We will also denote a clause $C_j$ as $\{L_j^1, L_j^2, L_j^3 \}$, with $L_j^i$ being a literal. Let us now construct an instance $\mathcal{I}$ of $\hiddenidentity$ that is positive if and only if $\varphi$ is satisfiable.

\paragraph{Description of $\mathcal{I}$.} Let us construct an election $E_{\varphi}$, which we will call an \emph{encoding} of $\varphi$. In $E_{\varphi}$ there is a set of $2|X| + 10|C|$ candidates. We will think of them as \emph{literal candidates} for the set $X \cup X'$ of variables, with $X' = \{x_{n+1}, x_{n + 5|C|} \}$. 
Then, for each $j \in [0, 2|X| + 10|C| -1]$, we refer to the candidate $c_j$ as $c_{x_j}$ if $j$ is even, and as $c_{\neg x_{j -1}}$ otherwise. Moreover, for each $j \in [0, |X| + 5|C| -1]$ we refer to $[j, j+1]$ as the \emph{sector} of $x_j$. Also, for a clause $C_i$, we say that variables $x_{|X|+5i +1}, x_{|X|+5i +2}, x_{|X|+5i +3}, x_{|X|+5i +4 }, x_{|X|+5i +5 }$ \emph{correspond} to $C_j$. We will also call them the first, second, third, fourth, and fifth variable corresponding to $C_j$ respectively.
Now, we consider $7|C| +2$ voters. We refer to $4|C|+2$ of them as \emph{main voters}. We will think of them as $2|C|+1$ pairs, i.e., $\{ p_1^1, p_1^2, \dots, p_{2|C|+1}^1, p_{2|C|+1}^2 \}$. 
Then, for every voter $p_i^1$ and a literal candidate $c_{x_j}$, let $rank_{p_i^1}(c_{x_j})=j$, while $rank_{p_i^1}(c_{\neg x_j})=j+1$. Similarly, for every voter $p_i^2$ and a literal candidate $c_{x_j}$, let $rank_{p_i^2}(c_{x_j})=j+1$, while $rank_{p_i^2}(c_{\neg x_j})=j$.

Now, for every clause $C_i$, we consider three voters, $v_{C_i}^1$, $v_{C_i}^2$, and $v_{C_i}^3$. Then, for each voter $v_{C_i}^j$ and a literal $x_k$ or $\neg x_k$ with $k \leq n-1$,  other than $\neg L_i^j$, let $rank_{v_{C_i}^j}(x_k)$ and $rank_{v_{C_i}^j}(\neg x_k)$ be smaller than a candidate $v_{l}$, where $l$ is a literal $x_l$ or $\neg x_l$, $l > k+1$, and $l$ is not in $C_j$. Moreover, let $rank_{v_{C_i}^j}(\neg L_i^j)$ be at the second position in the sector of the $x_{l+1}$, where $x_l$ is the variable corresponding to to $\neg L_i^j$. So, main voters rank candidates following the order of variables they correspond to, while for every variable $x_i$, a half of main voters ranks $c_{x_i}$ above $c_{\neg x_i}$, while the other half has the reverse order of those candidates. See the orders of $m_1$ and $m_2$ in Example \ref{ex:reduction} for an illustration of those votes. Also, let:

\begin{enumerate}
\item $v_{C_i}^1$ rank $c_{\neg x_{|X|+5i+2}}$ at the second position in the sector of $c_{x_{|X|+5i+3}}$ and $c_{x_{|X|+5i+4}}$ at the second position in the sector of $c_{x_{|X|+5i+5}}$.

\item $v_{C_i}^2$ rank $c_{ x_{|X|+5i+2}}$ at the second position in the sector of $c_{x_{|X|+5i+3}}$ and $c_{x_{|X|+5i+4}}$ at the second position in the sector of $c_{x_{|X|+5i+5}}$.

\item $v_{C_i}^3$ rank $c_{ x_{|X|+4i+2}}$ at the second position in the sector of $c_{x_{|X|+5i+3}}$ and $c_{\neg x_{|X|+5i+4}}$ at the second position in the sector of $c_{x_{|X|+5i+5}}$.
\end{enumerate}

Then, let remaining literal candidates are ranked by these voters with the natural order respecting indices.

Finally, let the target identity in $\mathcal{I}$ be an identity election with $5|C| + 2$ voters, as well as $|X| + 5|C|$, i.e., $|X'|$ candidates. We will refer to a subelection of $E_{\varphi}$ of this size as a \emph{potential solution}.
The following example shows an encoding of a simple formula.

\begin{example}\label{ex:reduction}
Take a formula $\varphi$ with $X=\{x_0,x_1,x_2\}$ and $C=C_0$, with $C_0=\{x_0, \neg x_1, x_2 \}$. Then, in the encoding of $\varphi$ we are given the set of candidates including $c_{x_0}, c_{x_0}, c_{x_0}, c_{x_0},  c_{x_0}, c_{x_0}, c_{\neg x_0}, c_{x_1}, c_{\neg x_1}, c_{x_2}$, \\ $c_{\neg x_2}, c_{x_3}, c_{\neg x_3}, c_{x_4}, c_{\neg x_4}, c_{x_5}, c_{\neg x_5}, c_{x_6}, c_{\neg x_6}$, $c_{x_7}$, $c_{\neg x_7}$.

Then, the voters' preferences are given by the following table.

\begin{table}[H]
\centering
\setlength{\tabcolsep}{1.2pt} 
\renewcommand{\arraystretch}{1.5} 
\scalebox{0.8}{

\begin{tabular}{lllllllllllllllll}
\hline
           & 1              & 2              & 3              & 4              & 5              & 6              & 7              & 8              & 9              & 10             & 11             & 12             & 13             & 14             & 15             & 16             \\ \hline
$m_1$      & $c_{x_0}$      & $c_{\neg x_0}$ & $c_{x_1}$      & $c_{\neg x_1}$ & $c_{x_2}$      & $c_{\neg x_2}$ & $c_{x_3}$      & $c_{\neg x_3}$ & $c_{x_4}$      & $c_{\neg x_4}$ & $c_{x_5}$      & $c_{\neg x_5}$ & $c_{x_6}$      & $c_{\neg x_6}$ & $c_{ x_7}$     & $c_{\neg x_7}$ \\
$m_2$      & $c_{\neg x_0}$ & $c_{x_0}$      & $c_{\neg x_1}$ & $c_{ x_1}$     & $c_{\neg x_2}$ & $c_{x_2}$      & $c_{\neg x_3}$ & $c_{ x_3}$     & $c_{\neg x_4}$ & $c_{ x_4}$     & $c_{\neg x_5}$ & $c_{x_5}$      & $c_{\neg x_6}$ & $c_{ x_6}$     & $c_{\neg x_7}$ & $c_{ x_7}$     \\
$x_0$      & $c_{x_0}$      & $c_{x_1}$      & $c_{\neg x_1}$ & $c_{\neg x_0}$ & $c_{x_2}$      & $c_{\neg x_2}$ & $c_{x_3}$      & $c_{\neg x_3}$ & $c_{x_4}$      & $c_{ x_5}$     & $c_{\neg x_5}$ & $c_{\neg x_4}$ & $c_{\neg x_6}$ & $c_{x_ 7}$     & $c_{\neg x_7}$ & $c_{x_6}$      \\
$\neg x_1$ & $c_{x_0}$      & $c_{\neg x_0}$ & $c_{\neg x_1}$ & $c_{x_2}$      & $c_{\neg x_2}$ & $c_{x_1}$      & $c_{x_3}$      & $c_{\neg x_3}$ & $c_{\neg x_4}$ & $c_{x_5}$      & $c_{\neg x_5}$ & $c_{x_4}$      & $c_{\neg x_6}$ & $c_{x_7}$      & $c_{\neg x_7}$ & $c_{x_6}$      \\
$x_2$      & $c_{x_0}$      & $c_{\neg x_0}$ & $c_{x_1}$      & $c_{\neg x_1}$ & $c_{x_2}$      & $c_{x_3}$      & $c_{\neg x_3}$ & $c_{\neg x_2}$ & $c_{\neg x_4}$ & $c_{x_5}$      & $c_{\neg x_5}$ & $c_{x_4}$      & $c_{\neg x_6}$ & $c_{x_7}$      & $c_{\neg x_7}$ & $c_{x_6}$      \\ \hline
\end{tabular}}
\caption{Example of an encoding of a formula with one clause $\{x_0, \neg x_1, x_2 \}$. The first two rows represent one of the pairs of main voters, while the other three the rankings of literal voters.}
\end{table}

\end{example}

\paragraph{Correctness of Construction.}

Let us first observe that every potential solution $s$ corresponds to some valuation over $X\cup X'$. To see that, suppose that for some variable $x_i$ both $c_{x_i}$ and $c_{\neg x_i}$ belong to $s$. Notice now that as there are $7|C|$ voters in $E_{\varphi}$, while $5|C|$ voters are included in $s$, by pigeonhole principle it holds that some pair of main voters $p_j^1, p_j^2$ is included in $s$. But then we have that $p_j^1, p_j^2$ do not have the same preference over $c_{x_i}$ and $c_{\neg x_i}$, so $s$ is not a potential solution. We will then say that a variable $x$ is \emph{true} in $s$ if $c_x$ belongs to $s$, and that it is \emph{false} otherwise.

Let us further observe that as a consequence it is the case that for every clause $C_j$ only one clause candidate is selected in $s$, as their preferences on variables corresponding to $C_j$ are pairwise inconsistent. This implies that exactly one of them needs to be selected in a potential solution.

Suppose now that $\varphi$ is satisfiable. Then, take a valuation $V$ over $X$ under which $\varphi$ is true. Now, take a subelection $s$ including all main voters and, for each clause $C_j$, one clause candidate corresponding to some literal true in this clause. Furthermore, for each variable $x_i \in X$, let $x_i$ be true in $s$ if $x_i$ is true in $V$ and let it be false otherwise. Then, for each clause $C_j$, let the first, third, and fifth variable corresponding to $C_j$ be true in $s$. Also, let the second variable corresponding to $C_j$ be true in $s$ if $v_{C_i}^1$ is selected in $s$ and let it be false in $s$ otherwise. Finally, let the fourth variable corresponding to $C_j$ be true in $s$ if $v_{C_i}^1$ is selected in $s$ and let it be false in $s$ otherwise. Then, by construction, we get that $s$ is an identity election.

Suppose that $\varphi$ is not satisfiable and suppose towards contradiction that there exists a potential solution $s$ that is an identity election. Then, by previous observations, we need to have that exactly one of $c_{x_i}$, $c_{\neg x_i}$ is selected in $s$, for every variable $x_i \in X \cup X'$. Furthermore, as $\varphi$ is not satisfiable, we need to have that for some clause $C_j$ $L_j^1, L_j^2, L_j^3$ are false in $s$. Then, without loss of generality, suppose that $v_{C_i}^1$ is selected in $s$. Observe that $v_{C_i}^1$ ranks $c_{\neg L_j^1}$ lower than both candidates corresponding to a variable with a lower index. But then, since $L_j^1$ is false in $s$, $c_{\neg L_j^1}$ is selected in this potential solution. This implies, however, that $s$ is not an identity election, as $c_{\neg L_j^1}$ is conflicting main voters' rankings.
\end{proof}

\section{Missing Computational Results for \textsc{Hidden-ID}} \label{appendix:hidden-identity-computational-results}

\paragraph{Results for Counting Problems for $\hiddenidentity$}

\begin{lemma} \label{thm:hidden-identity-voters-number-of-identity-candidate-subsets-p}
    Computation of the number of size-$m'$ candidates that form an identity for a given a set of voters can be done in polynomial time.
\end{lemma}
\begin{proof}
    Suppose we are given a $\hiddenidentity$ instance $(E, m', n')$ and a set of $n'$ voters $V' = \{v_{i_1}, v_{i_2}, \dots, v_{i_{n'}}\}$. We ask what is the number of size-$m'$ subsets of candidates that are preferred exactly in the same order by all voters from $V'$. Equivalently, we ask for the number of the Common Subsequences (CS) of their votes of size exactly $m'$.
    
    We construct graph $G$ of majority relations between candidates. Specifically, let vertices be the candidates and we draw a directed edge from $c_{y_1}$ to $c_{y_2}$ if and only if $c_{y_1}$ is preferred over $c_{y_2}$ by all voters from $V'$. The constructed graph is naturally a directed acyclic graph (DAG), for which we ask for the number of size-$m'$ sequences of vertices such that there is a top-down path connecting all these vertices (not necessary directly). For the sake of brevity, by $\textit{isAncestor}(p,v)$ we mean that there exists a path from $p$ to $v$ \footnote{In our case, since all voters' preference relation is transitive, $\textit{isAncestor}(p,v)$ holds if and only if there is an edge $p \rightarrow v$.}.
    We can solve it with a dynamic programming as follows: Let $DP[v,j]$ be the number of size-$j$ sequences of voters such that all these voters are connected by a top-down path (not necessary directly) and $v$ is the last vertex in the sequence. We sort vertices topologically and process them in this order. If a vertex $v$ does not have any incoming edges, then $\textit{DP}[v,j] = [j=1]$. Otherwise, $DP[v,j] = \sum_{p | \textit{isAncestor}(p,v)}{\textit{DP}[p,j-1]}$. Finally, we return $\sum_{v \in V(G)} {\textit{DP}[v,m']}$. 

    The time complexity of this algorithm is $O(n' \cdot m^2 + m^2 + m^2) = O(n' \cdot m^2)$.
\end{proof}

\begin{theorem} \label{thm:sharp-hidden-identity-fpt-n}
    $\hiddenidentity$~and~$\sharphiddenidentity$ are $\fpt$ parameterized by the number of voters  and $\xp$ parameterized by the number of voters in the solution ($n'$).
\end{theorem}
\begin{proof}
    Suppose we are given a $\sharphiddenidentity$ instance $(E, m', n')$. We iterate through all possible size-$n'$ subsets of voters and use the algorithm provided in the proof of Lemma \ref{thm:hidden-identity-voters-number-of-identity-candidate-subsets-p} to determine the number of its identities. The sum of the number of identities for all size-$n'$ subsets of voters is the answer.

    For the decisive variant, we accept if the answer is greater than $0$, otherwise we reject. For a better performance, we can iterate through all possible size-$n'$ sets of voters and accept if the algorithm \ref{thm:hidden-identity-voters-verification-p} found any identity, otherwise we reject.
\end{proof}

We can adapt the approach from Theorem~\ref{thm:hidden-identity-candidates-verification-p} to compute the number of identities for a given set and order of $m'$ candidates.

\begin{lemma} \label{thm:hidden-identity-candidates-number-of-identity-voter-subsets-p}
    Computation of the number of size-$n'$ voters that form an identity for a given a set and order of candidates can be done in polynomial time.
\end{lemma}
\begin{proof}
    Suppose we are given a $\hiddenidentity$ instance $(E, m', n')$ and a set of $m'$ candidates $C' = \{c_{i_1}, c_{i_2}, \dots, c_{i_{m'}}\}$. We ask how many size-$n'$ sets of voters that rank these candidates exactly in the same order. 

    We define and compute dictionary $D$ in exactly the same way as in algorithm \ref{thm:hidden-identity-candidates-verification-p}. Then, we return $\sum_{k \in D} {|D[k]| \choose n'}$.
    For a given order $ord(C')$ of $m'$ candidates, we just return $D[ord(C')] \choose n'$ instead of the above sum.
    In both cases, time complexity is $O(n(m + m')) = O(nm)$.
\end{proof}

Thanks to Lemma \ref{thm:hidden-identity-candidates-number-of-identity-voter-subsets-p} we obtain parameterized tractability of $\hiddenidentity$.

\begin{theorem} \label{thm:sharp-hidden-identity-fpt-m}
    $\hiddenidentity$~and~$\sharphiddenidentity$ are~$\fpt$ parameterized by the number of candidates and~$\xp$ 
    by the number of candidates in the solution ($m'$).
\end{theorem}
\begin{proof}
    Suppose we are given a $\sharphiddenidentity$ instance $(E, m', n')$. We iterate through all possible size-$m'$ sets of candidates and use the algorithm \ref{thm:hidden-identity-candidates-number-of-identity-voter-subsets-p} to determine the number of its identities. The sum of the number of identities for all size-$m'$ sets of candidates is the answer.
    
    For the decisive variant, we accept if the answer is greater than $0$, otherwise we reject. For a better performance, we can iterate through all possible size-$m'$ sets of candidates and accept if the algorithm \ref{thm:hidden-identity-candidates-verification-p} found any identity, otherwise we reject.
\end{proof}

\section{Missing Hardness Proof for \textsc{Hidden-AN}}

\paragraph{Proof of Theorem \ref{thm:IHANP}.}

\begin{proof}

We will show that $\hiddenantagonism$ is \np-hard by reduction from $\threesat$. Consider a 3-CNF formula $\varphi$ with the set of variables $X=\{x_0, \dots, x_n \}$ and the set of clauses $C= \{C_0, \dots, C_m\}$. We will also denote a clause $C_j$ as $\{L_j^1, L_j^2, L_j^3 \}$, with $L_j^i$ being a literal. Let us construct an instance $\mathcal{I}$ of $\hiddenantagonism$ that is positive if and only if $\varphi$ is satisfiable.

\paragraph{Description of $\mathcal{I}$.} Let us construct an election $E_{\varphi}$, which we will call an \emph{encoding} of $\varphi$. 
In $E_{\varphi}$ there is a set of $2|X| + 10|C| +4$ candidates. We will think of them as \emph{literal candidates} for the set $X \cup X'$ of variables, with $X' = \{x_{n+1}, \dots, x_{n + 5|C|+2} \}$. 
Then, for each $j \in [0, 2n + 8|C| -1]$, we refer to the candidate $c_j$ as $c_{x_j}$ if $j$ is even, and as $c_{\neg x_{j -1}}$ otherwise. Moreover, for each $j \in [0, 2n + 4|C| -1]$ we refer to $[j, j+1]$ as the \emph{sector} of $x_j$. Also, for a clause $C_i$, we say that variables $x_{|X|+5i +1}, x_{|X|+5i +2}, x_{|X|+5i +3}, x_{|X|+5i +4 }, x_{|X|+5i +5 }$ \emph{correspond} to $C_j$. We will also call them first, second, third, fourth and fifth variable corresponding to $C_j$ respectively.

Now, we consider $8|C|$ voters. Let $4|C|$ of them be \emph{main voters}. We will think of them as $2|C|$ pairs, i.e., $\{ p_1^1, p_1^2, \dots, p_{4|C|}^1, p_{2|C|}^2 \}$. 
Let us first consider the group of $2|C|$ \emph{type A} main voters. Then, for every voter $p_i^1$ and a literal candidate $c_{x_j}$, with $j\leq |X'|-3$, let $rank_{p_i^1}(c_{x_j})=2j+1$, while $rank_{p_i^1}(c_{\neg x_j})=2j+2$. Similarly, for every voter $p_i^2$ and a literal candidate $c_{x_j}$, let $rank_{p_i^2}(c_{x_j})=j+1$, while $rank_{p_i^2}(c_{\neg x_j})=j$. Moreover, let half of voters in each of these groups rank $c_{ x_{|X'|-2}} \succ c_{ x_{|X'|-1}} \succ c_{ \neg x_{|X'|-1}} \succ c_{ \neg x_{|X'|-2}}$, and the other half $c_{ \neg x_{|X'|-2}} \succ c_{ x_{|X'|-1}} \succ c_{ \neg x_{|X'|-1}} \succ c_{  x_{|X'|-2}}$. We assume that these halfs are distributed equally among the groups within this type.

Then, we take the group of $2|C|$ \emph{type B} main voters. Let each of them rank the candidate $c_{x_j}$, with $j\leq |X \cup X'|-3$, at position $2|X| -2j +5$, and $c_{x_j}$ at position $2|X| -2j +5$. In other words, type B main voters have a reversed order of candidates corresponding to particular variables, compared to type A main voters. Moreover, let half of them rank $c_{x_j} \succ c_{\neg x_j}$, for each $j\leq |X \cup X'|-3$, while the other half rank $c_{\neg x_j} \succ c_{ x_j}$, for each $j\leq |X'|-3$. Furthermore, let a half of them rank $c_{ x_{|X \cup X'|-2}} \succ c_{ x_{|X'|-1}} \succ c_{ \neg x_{|X'|-1}} \succ c_{ \neg x_{|X \cup X'|-2}}$ first, while a half of them start their ranking with $c_{ x_{|X'|-2}} \succ c_{ x_{|X\cup X'|-1}} \succ c_{ \neg x_{|X \cup X'|-1}} \succ c_{ \neg x_{|X \cup X'|-2}}$. Again, we assume that these halfs are distributed equally among the groups within this type.

Now, for every clause $C_i$, we consider two groups of three voters, each of them consisting of voters $v_{C_i}^1$, $v_{C_i}^2$, and $v_{C_i}^3$. 
Then, for each clause $C_i$ of, a voter $v_{C_i}^j$ and a literal $x_k$ or $\neg x_k$ with $k \leq n-1$,  other than $\neg L_i^j$, let $rank_{v_{C_i}^j}(x_k)$ and $rank_{v_{C_i}^j}(\neg x_k)$ be smaller than a candidate $v_{l}$, where $l$ is a literal $x_l$ or $\neg x_l$, $l > k+1$, and $l$ is not in $C_j$. Moreover, let $rank_{v_{C_i}^j}(\neg L_i^j)$ be at the second position in the sector of the $x_{l+1}$, where $x_l$ is the variable corresponding to $\neg L_i^j$. Also, for one of the groups of such clause voters, let:

\begin{enumerate}
\item $v_{C_i}^1$ rank $c_{\neg x_{|X|+5i+2}}$ at the second position in the sector of $c_{x_{|X|+5i+3}}$ and $c_{x_{|X|+5i+4}}$ at the second position in the sector of $c_{x_{|X|+5i+5}}$.

\item $v_{C_i}^2$ rank $c_{ x_{|X|+5i+2}}$ at the second position in the sector of $c_{x_{|X|+5i+3}}$ and $c_{x_{|X|+5i+4}}$ at the second position in the sector of $c_{x_{|X|+5i+5}}$.

\item $v_{C_i}^3$ rank $c_{ x_{|X|+4i+2}}$ at the second position in the sector of $c_{x_{|X|+5i+3}}$ and $c_{\neg x_{|X|+5i+4}}$ at the second position in the sector of $c_{x_{|X|+5i+5}}$.
\end{enumerate}

Then, let remaining literal candidates are ranked by one group of those voters with the natural order respecting indices.

Furthermore, for the other group of such clause voters, let:

\begin{enumerate}
\item $v_{C_i}^1$ rank $c_{\neg x_{|X|+5i+2}}$ lower than  $c_{x_{|X|+5i+3}}$ and $c_{x_{|X|+5i+4}}$ lower than $c_{x_{|X|+5i+5}}$.

\item $v_{C_i}^2$ rank $c_{ x_{|X|+5i+2}}$  lower than  $c_{x_{|X|+5i+3}}$ and $c_{x_{|X|+5i+4}}$ lower than $c_{x_{|X|+5i+5}}$.

\item $v_{C_i}^3$ rank $c_{x_{|X|+5i+2}}$ lower than  $c_{x_{|X|+5i+3}}$ and $c_{\neg x_{|X|+5i+4}}$ lower than $c_{x_{|X|+5i+5}}$.
\end{enumerate}

Then, let remaining literal candidates are ranked by one group of those voters with the natural, reversed order respecting indices.

Finally, let the target identity in $\mathcal{I}$ be an antagonism election with $4|C|$ voters, as well as $|X'|$ candidates. We will refer to a subelection of $E_{\varphi}$ of this size as a \emph{potential solution}.

Below we provide of an encoding for a simple instance.

\begin{example}\label{ex:antagonismhard}
Take a formula $\varphi$ with $X=\{x_0,x_1,x_2\}$ and $C=C_0$, with $C_0=\{x_0, \neg x_1, x_2 \}$. Then, in the encoding of $\varphi$ we are given the set of candidates including: $c_{x_0}, c_{x_0}, c_{x_0}, c_{x_0},  c_{x_0}, c_{x_0}, c_{\neg x_0}, c_{x_1}, c_{\neg x_1}, c_{x_2}$, $c_{\neg x_2}, c_{x_3}, c_{\neg x_3}, c_{x_4}, c_{\neg x_4}, c_{x_5}, c_{\neg x_5}, c_{x_6}, c_{\neg x_6}$, $c_{x_7}$, $c_{\neg x_7}$, $c_{x_8}$, $c_{\neg x_8}$, $c_{x_9}$, and $c_{\neg x_9}$.

\begin{table}[H]
\setlength{\tabcolsep}{1pt} 
\renewcommand{\arraystretch}{1.5}
\centering
\scalebox{0.7}{

\begin{tabular}{@{}lllllllllllllllllllll@{}}
\toprule

                        & 1              & 2                       & 3              & 4              & 5              & 6              & 7              & 8              & 9              & 10             & 11             & 12             & 13             & 14             & 15             & 16             & 17             & 18             & 19                                & 20                                 \\ \midrule
$m_1^1$                 & $c_{x_0}$      & $c_{\neg x_0}$          & $c_{x_1}$      & $c_{\neg x_1}$ & $c_{x_2}$      & $c_{\neg x_2}$ & $c_{x_3}$      & $c_{\neg x_3}$ & $c_{x_4}$      & $c_{\neg x_4}$ & $c_{x_5}$      & $c_{\neg x_5}$ & $c_{x_6}$      & $c_{\neg x_6}$ & $c_{ x_7}$     & $c_{\neg x_7}$ & $c_{\neg x_8}$ & $c_{ x_9}$     & $c_{\neg x_9}$                    & $c_{x_8}$                          \\
$m_2^1$                 & $c_{\neg x_0}$ & $c_{x_0}$               & $c_{\neg x_1}$ & $c_{ x_1}$     & $c_{\neg x_2}$ & $c_{x_2}$      & $c_{\neg x_3}$ & $c_{ x_3}$     & $c_{\neg x_4}$ & $c_{ x_4}$     & $c_{\neg x_5}$ & $c_{x_5}$      & $c_{\neg x_6}$ & $c_{ x_6}$     & $c_{\neg x_7}$ & $c_{ x_7}$     & $c_{ x_8}$     & $c_{\neg x_9}$ & $c_{ x_9}$                        & $c_{\neg x_8}$                     \\
$m_1^2$ & $c_{x_8}$      & $c_{\neg x_9}$          & $c_{ x_9}$     & $c_{\neg x_8}$ & $c_{x_7}$      & $c_{\neg x_7}$ & $c_{x_6}$      & $c_{\neg x_6}$ & $c_{x_5}$      & $c_{\neg x_5}$ & $c_{ x_4}$     & $c_{\neg x_4}$ & $c_{ x_3}$     & $c_{\neg x_3}$ & $c_{ x_2}$     & $c_{\neg x_2}$ & $c_{x_1}$      & $c_{\neg x_1}$ & $c_{x_0}$                         & $c_{\neg  x_0}$ \\
$m_2^1$ & $c_{\neg x_8}$ & $c_{\neg x_8}$ & $c_{x_9}$      & $c_{x_8}$      & $c_{\neg x_7}$ & $c_{x_7}$      & $c_{\neg x_6}$ & $c_{x_6}$      & $c_{\neg x_5}$ & $c_{x_5}$      & $c_{\neg x_4}$ & $c_{x_4}$      & $c_{\neg x_3}$ & $c_{x_3}$      & $c_{\neg x_2}$ & $c_{x_2}$      & $c_{\neg x_1}$ & $c_{x_1}$      & $c_{\neg x_0}$                    & $c_{x_0}$                          \\
$x_0^1$                 & $c_{ x_9}$     & $c_{\neg x_9}$          & $c_{ x_8}$     & $c_{\neg x_8}$ & $c_{\neg x_6}$ & $c_{x_7}$      & $c_{\neg x_7}$ & $c_{ x_6}$     & $c_{x_4}$      & $c_{x_5}$      & $c_{\neg x_5}$ & $c_{\neg x_4}$ & $c_{ x_3}$     & $c_{\neg x_3}$ & $c_{ x_2}$     & $c_{\neg x_2}$ & $c_{\neg x_0}$ & $c_{ x_1}$     & $c_{ \neg x_1}$ & $c_{x_0}$                          \\
$\neg x_1^1$            & $c_{ x_9}$     & $c_{\neg x_9}$          & $c_{ x_8}$     & $c_{\neg x_8}$ & $c_{x_6}$      & $c_{ x_7}$     & $c_{\neg x_7}$ & $c_{\neg x_6}$ & $c_{x_4}$      & $c_{x_5}$      & $c_{\neg x_5}$ & $c_{\neg x_4}$ & $c_{x_3}$      & $c_{\neg x_3}$ & $c_{ x_1}$     & $c_{ x_2}$     & $c_{\neg x_2}$ & $c_{\neg x_1}$ & $c_{ x_0}$                        & $c_{\neg x_0}$                     \\
$x_2^1$                 & $c_{ x_9}$     & $c_{\neg x_9}$          & $c_{ x_8}$     & $c_{\neg x_8}$ & $c_{x_6}$      & $c_{ x_7}$     & $c_{\neg x_7}$ & $c_{\neg x_6}$ & $c_{x_4}$      & $c_{x_5}$      & $c_{\neg x_5}$ & $c_{\neg x_4}$ & $c_{\neg x_2}$ & $c_{ x_3}$     & $c_{\neg x_3}$ & $c_{ x_2}$     & $c_{ x_1}$     & $c_{\neg x_1}$ & $c_{ x_0}$                        & $c_{\neg x_0}$                     \\
$x_0^2$                 & $c_{x_0}$      & $c_{x_1}$               & $c_{\neg x_1}$ & $c_{\neg x_0}$ & $c_{x_2}$      & $c_{\neg x_2}$ & $c_{x_3}$      & $c_{\neg x_3}$ & $c_{x_4}$      & $c_{ x_5}$     & $c_{\neg x_5}$ & $c_{\neg x_4}$ & $c_{\neg x_6}$ & $c_{x_ 7}$     & $c_{\neg x_7}$ & $c_{x_6}$      & $c_{x_8}$      & $c_{\neg x_8}$ & $c_{x_9}$                         & $c_{\neg x_9}$                     \\
$\neg x_1^2$            & $c_{x_0}$      & $c_{\neg x_0}$          & $c_{\neg x_1}$ & $c_{x_2}$      & $c_{\neg x_2}$ & $c_{x_1}$      & $c_{x_3}$      & $c_{\neg x_3}$ & $c_{\neg x_4}$ & $c_{x_5}$      & $c_{\neg x_5}$ & $c_{x_4}$      & $c_{\neg x_6}$ & $c_{x_7}$      & $c_{\neg x_7}$ & $c_{x_6}$      & $c_{x_8}$      & $c_{\neg x_8}$ & $c_{x_9}$                         & $c_{\neg x_9}$                     \\
$x_2^2$                 & $c_{x_0}$      & $c_{\neg x_0}$          & $c_{x_1}$      & $c_{\neg x_1}$ & $c_{x_2}$      & $c_{x_3}$      & $c_{\neg x_3}$ & $c_{\neg x_2}$ & $c_{\neg x_4}$ & $c_{x_5}$      & $c_{\neg x_5}$ & $c_{x_4}$      & $c_{\neg x_6}$ & $c_{x_7}$      & $c_{\neg x_7}$ & $c_{x_6}$      & $c_{x_8}$      & $c_{\neg x_8}$ & $c_{x_9}$                         & $c_{\neg x_9}$                     \\ \bottomrule
\end{tabular}}

\caption{Example of an encoding of the formula $\{x_0, \neg x_1, x_2\}$. The first four votes represent the main voters.}
\end{table}

Observe how selecting literal candidates corresponding to $x_0, \neg x_1$, and $x_2$ allows us to construct an antagonism in this instance.
\end{example}

\paragraph{Correctness of Construction.}

Let us first observe that every potential solution $s$ corresponds to some valuation over $X\cup X'$. To see that, suppose that for some variable $x_i$, both $c_{x_i}$ and $c_{\neg x_i}$ belong to $s$. Notice now that as there are $4|C|$ voters in $E_{\varphi}$, while $4|C|$ voters are included in $s$, by pigeonhole principle it holds that some pair of main voters belonging to the same type, i.e., $p_j^1, p_j^2$, is included in $s$. But then, we have that $p_j^1, p_j^2$ do not have the same preference over $c_{x_i}$ and $c_{\neg x_i}$, so $s$ is not a potential solution. We will then say that a variable $x$ is \emph{true} in $s$ if $c_x$ is in $s$, and that it is \emph{false} otherwise.

Let us further observe that as a consequence it is the case that for every clause $C_j$ only one clause voter is selected in $s$, for each type of orders,  as their preferences on variables corresponding to $C_j$ are pairwise inconsistent. This implies that exactly one of them needs to be selected in a potential solution. Hence, at most $2|C|$ clause voters.

Suppose now that $\varphi$ is satisfiable. Then, take a valuation $V$ over $X$ under which $\varphi$ is true. Then, take a subelection $s$ including half of main voters of type A and of type B, with a consistent order of literal candidates for the last two variables in $X'$. Then, for each clause $C_j$, let us one clause candidate corresponding to some literal true in this clause, for each type of order. Moreover, for each variable $x_i \in X$, let $x_i$ be true in $s$ if $x_i$ is true in $V$ and let it be false otherwise. Furthermore, for each clause $C_j$, let the first, third, and fifth variable corresponding to $C_j$ be true in $s$. Also, let the second variable corresponding to $C_j$ be true in $s$ if $v_{C_i}^1$ is selected in $s$ and let it be false in $s$ otherwise. Finally, let the fourth variable corresponding to $C_j$ be true in $s$ if $v_{C_i}^1$ is selected in $s$ and let it be false in $s$ otherwise. Then, by construction, we get that $s$ is an antagonism.

Suppose now that $\varphi$ is not satisfiable and suppose towards contradiction that there exists a potential solution $s$ which is an antagonism. Then, by previous observation, we need to have that exactly one of $c_{x_i}$, $c_{\neg x_i}$ is selected in $s$, for every variable $x_i \in X \cup X'$. Observe that then the clause voters form an antagonism of size $2|C|$.  Hence, at least $2|C|$ main voters are selected in $s$. Moreover, as $s$ is an antagonism, at most half of main voters from type A and from type B are selected in $s$, as otherwise they would not form identities of equal sizes.
Furthermore, as $\varphi$ is not satisfiable, we need to have that for some clause $C_j$ $L_j^1, L_j^2, L_j^3$ are false in $s$. Then, without loss of generality, suppose that two versions $v_{C_i}^1$ are selected in $s$. Observe that $v_{C_i}^1$ ranks $c_{\neg L_j^1}$ lower than both candidates corresponding to a variable with a lower index. But then, since $L_j^1$ is false in $s$, $c_{\neg L_j^1}$ is selected in this potential solution. This implies, however, that $s$ is not an antagonism.

\end{proof}

\section{Missing Computational Results for \textsc{Hidden-AN}} \label{appendix:hidden-antagonism-computational-results}

\paragraph{Proof of Theorem \ref{thm:hidden-antagonism-ftp-n-and-fpt-m}}

\begin{proof}
    Suppose we are given a $\hiddenantagonism$ instance $(E, m', n')$. 
    
    For the parameter number of voters ($n$), we iterate through all possible size-$n'$ subsets of voters and check if the algorithm \ref{thm:hidden-antagonism-voters-verification-p} found any antagonism subelection consisting of at least $m'$ candidates, if so, then we accept, otherwise we reject.

    For the parameter number of candidates ($m$), we iterate through all size-$m'$ sets of candidates and accept if the algorithm \ref{thm:hidden-antagonism-candidates-verification-p} found any identity subelection consisting of at least $n'$ voters, otherwise we reject.
\end{proof}

\paragraph{Proof of Proposition \ref{prop:ilphiddenant}.}

\begin{proof}
We solve this problem using similar ILP to the one solving \hiddenidentity.
We need to divide voters in two groups.
For each~$i \in [n]$, beside $V_i$ we define 
  a variable~$U_{i}$ with the intention that value~$1$ indicates that
  voter~$v_i$ is selected as a member of the second group. Variables $S$ and $C$ have the same meaning as before for the \hiddenidentity. Finally, beside variable $P$ we introduce variable $R$ with the same meaning, but dedicated to the voters from the second group.

  \begin{align}
    \label{ilp:a1}
    &\textstyle\sum_{i \in [n]} V_{i} = \nicefrac{n'}{2}, \\
    \label{ilp:a2}
    &\textstyle\sum_{i \in [n]} U_{i} = \nicefrac{n'}{2}, \\
    \label{ilp:a3}
    &\textstyle\sum_{j \in [m]} C_{j} = m', \\
    \label{ilp:a4}
    &S_{j_1,j_2} + S_{j_2,j_1} = C_{j_1} \cdot C_{j_2}, \text{\quad}\forall_{j_1,j_2 \in [m]}, \\
    \label{ilp:a5}
    &P_{i,j_1,j_2} = V_i\cdot S_{j_1,j_2}, \text{\quad}\forall_{i \in [n], j_1,j_2 \in [m]}, \\
    \label{ilp:a6}
    &R_{i,j_1,j_2} = U_i\cdot S_{j_2,j_1}, \text{\quad}\forall_{i \in [n], j_1,j_2 \in [m]}.
  \end{align}
  
Constraints~\eqref{ilp:a1},~\eqref{ilp:a2} and~\eqref{ilp:a3} ensure that we select proper numbers of voters and candidates.
   Constraints~\eqref{ilp:a4},~\eqref{ilp:a5} and~\eqref{ilp:a6} implements the logic of $S$, $P$, and $R$ variables, respectively.
  The optimization goal is to minimize: \\
    \begin{align*}
  \textstyle\sum_{i \in[n],\ j_1, j_2 \in [m]} P_{i,j_1,j_2} \cdot
  W_{i,j_1,j_2} + R_{i,j_1,j_2} \cdot W_{i,j_1,j_2}
  \end{align*}

where $W_{i,j_1,j_2} = [pos_{v_i}(c_{j_1}) > pos_{v_i}(c_{j_2})]$.
  
\end{proof}





\section{Missing Experimental Results}

Here, we present maps for {$\maxclone(E,3)$}, {$\maxclone(E,4)$}, and {$\maxclone(E,5)$}, i.e., {\maxclone} with three, four, and five candidates, respectively.
In~\Cref{exp:maps_an_comparison} we demonstrate comparison of maps for three different notions of antagonism. Moreover, in~\Cref{tab:pcc_an_comparison} we provide PCC between three different variants of antagonism, and the distances from $AN$ and from $ID$.
Then, we report numbers of elections from each statistical culture in our dataset.
Further, we present average values for $\maxclone(E,2)$, $\maxid(E,5)$, and $\maxan(E,5)$ for each statistical culture separately.

\begin{figure*}
  \begin{subfigure}[b]{0.33\textwidth}
      \centering
      \includegraphics[width=6.cm, trim={0.2cm 0.2cm 0.2cm 0.2cm}, clip]{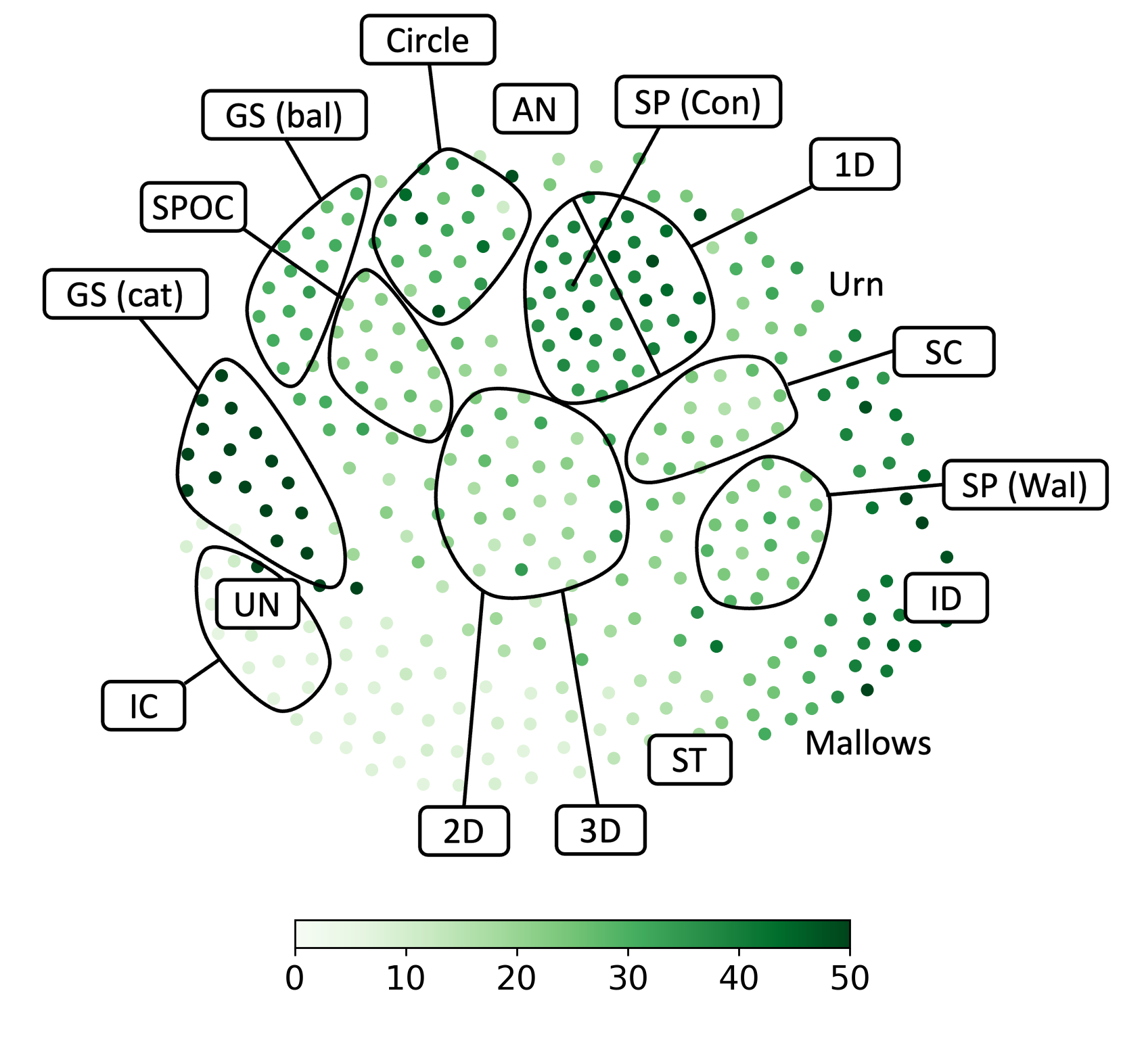}
      \caption{ $\maxclone(E,3)$}
  \end{subfigure}%
    \begin{subfigure}[b]{0.33\textwidth}
      \centering
      \includegraphics[width=6.cm, trim={0.2cm 0.2cm 0.2cm 0.2cm}, clip]{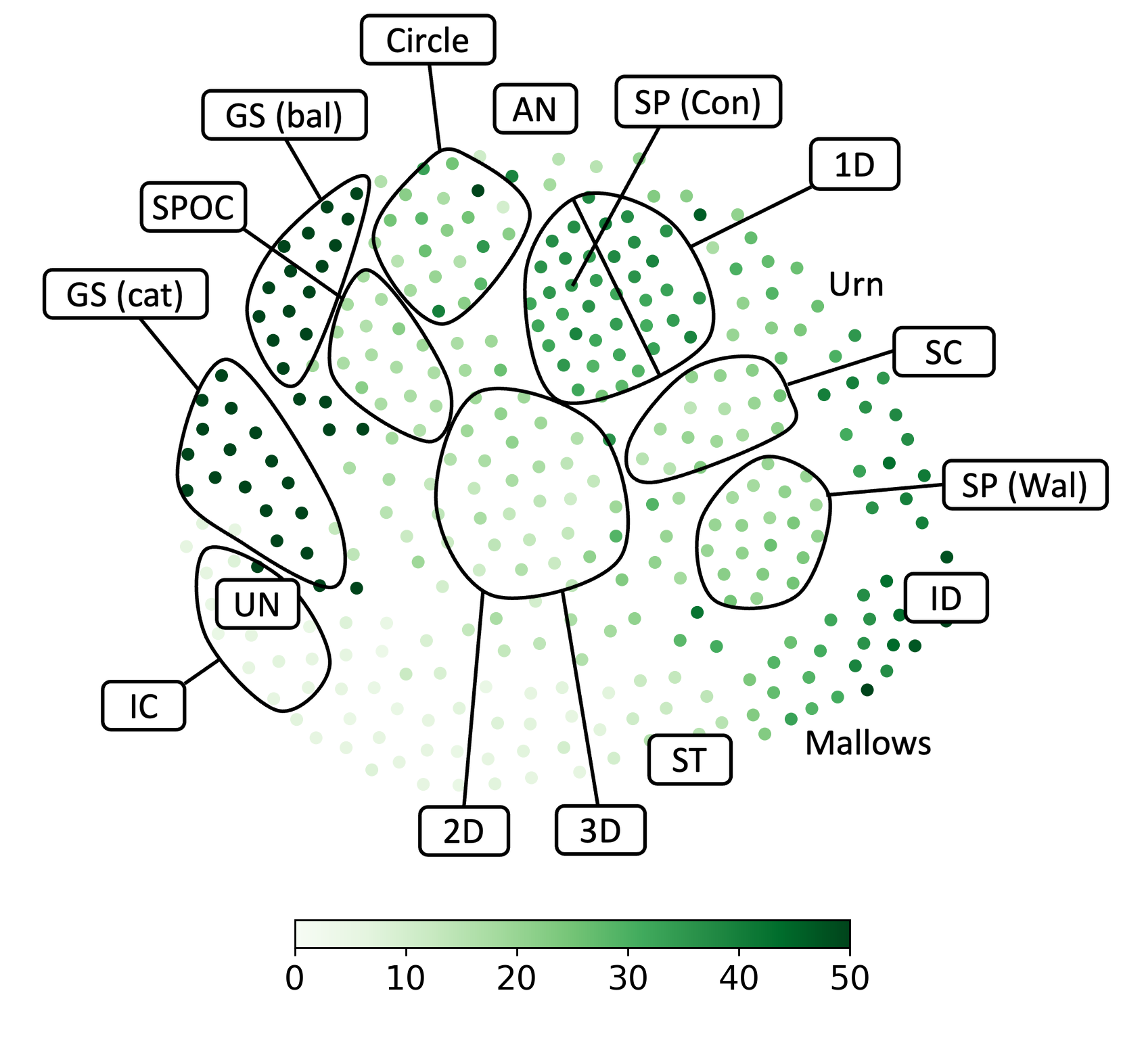}
      \caption{ $\maxclone(E,4)$}
  \end{subfigure}%
    \begin{subfigure}[b]{0.33\textwidth}
      \centering
      \includegraphics[width=6.cm, trim={0.2cm 0.2cm 0.2cm 0.2cm}, clip]{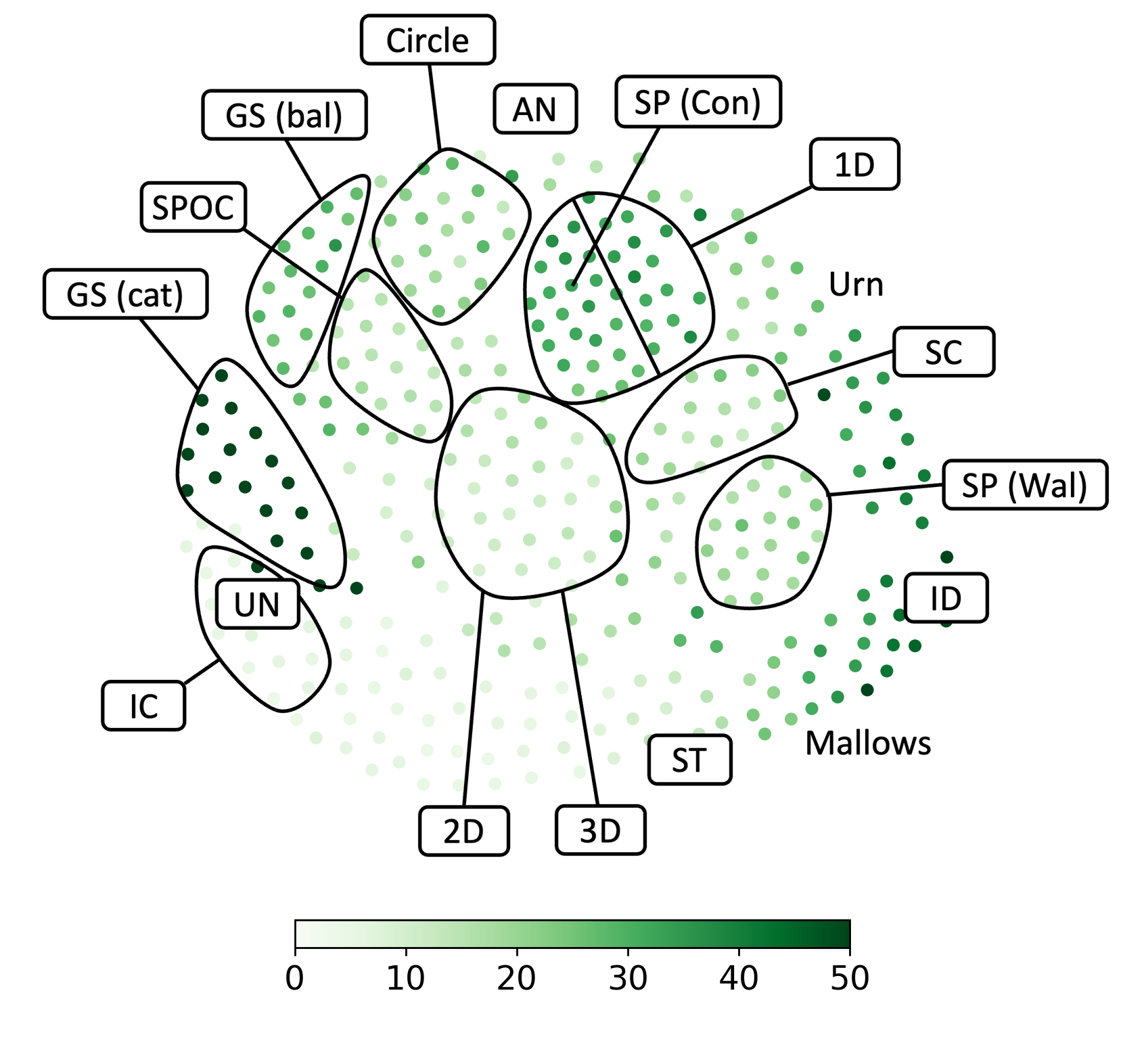}
      \caption{$\maxclone(E,5)$}
  \end{subfigure}\label{exp:maps_clones_extra}

   \caption{ Maps of elections with 10 candidates and 50 voters. Each point represents a single election. The darker the point is, the more voters agree on a certain set of candidates being clones. On each map, ID label marks the identity election, and AN label marks the antagonism election. }
\end{figure*}

\begin{figure*}
  \begin{subfigure}[b]{0.33\textwidth}
      \centering
      \includegraphics[width=6.cm, trim={0.2cm 0.2cm 0.2cm 0.2cm}, clip]{img/max/iha_max_5_map.png}
      \caption{ \maxan(E,5) \\ ``Default'' variant }
  \end{subfigure}%
    \begin{subfigure}[b]{0.33\textwidth}
      \centering
      \includegraphics[width=6.cm, trim={0.2cm 0.2cm 0.2cm 0.2cm}, clip]{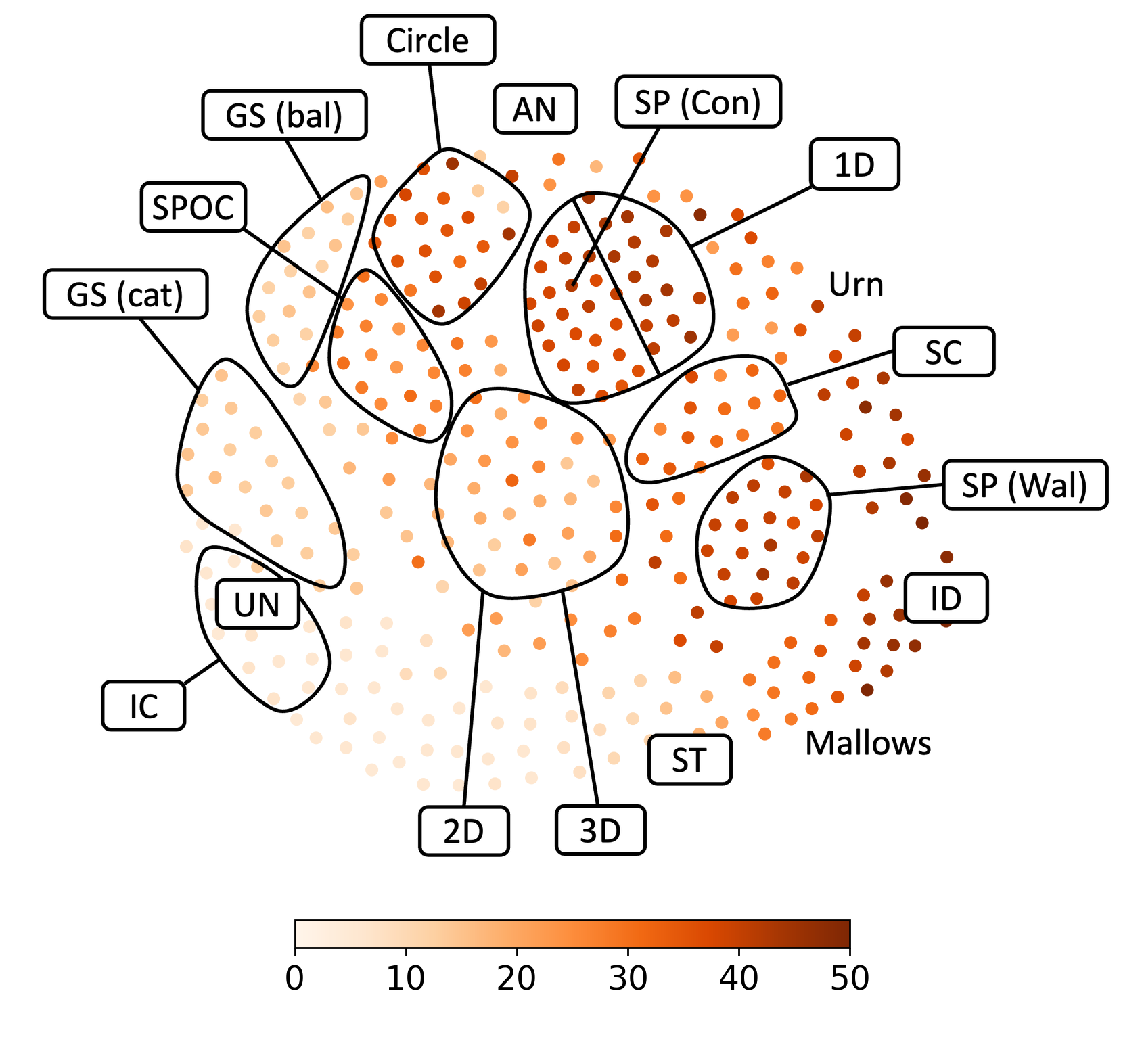}
      \caption{\maxan(E,5) \\ ``Sum'' variant 
       }
  \end{subfigure}%
    \begin{subfigure}[b]{0.33\textwidth}
      \centering
      \includegraphics[width=6.cm, trim={0.2cm 0.2cm 0.2cm 0.2cm}, clip]{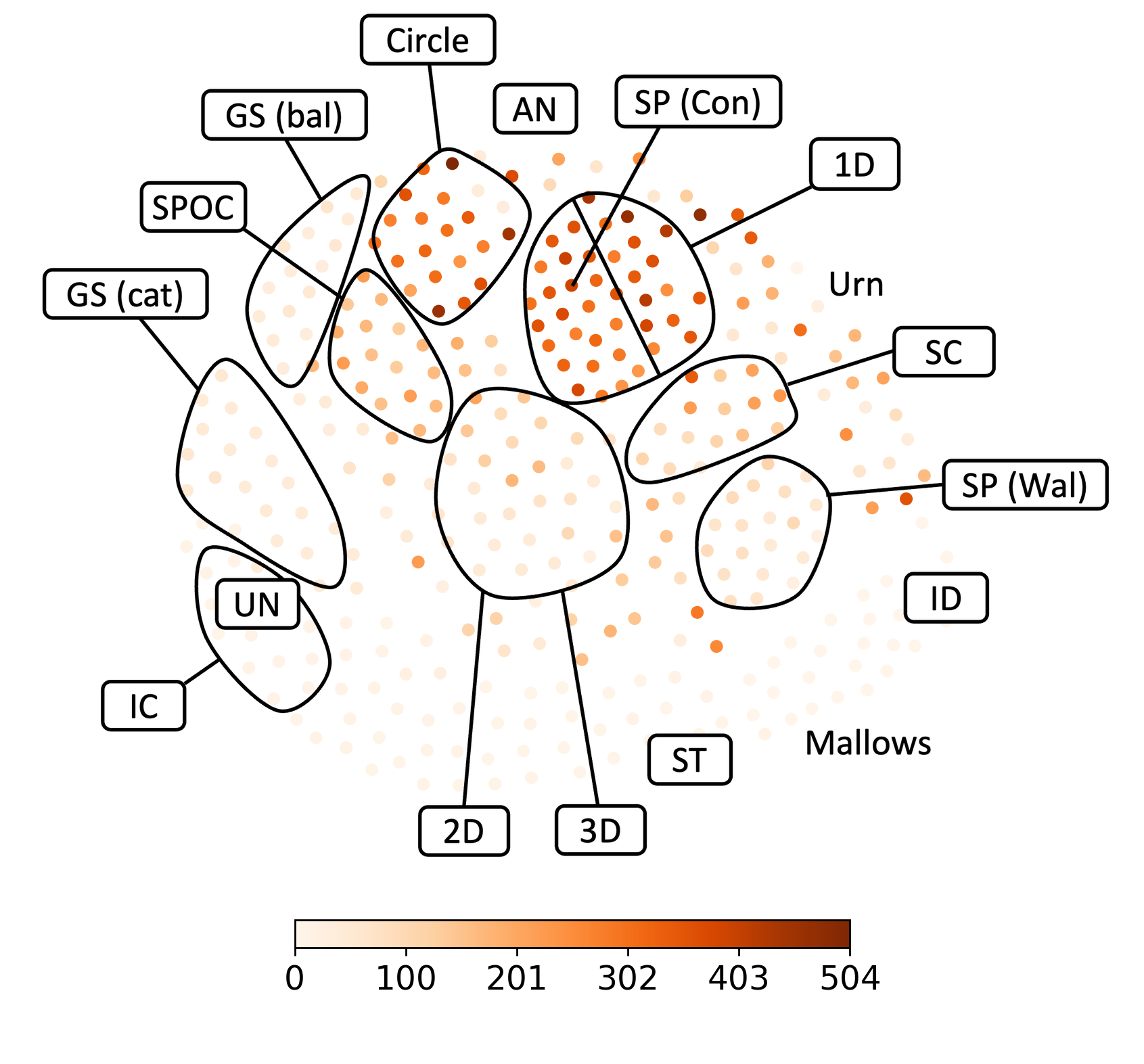}
      \caption{\maxan(E,5) \\ ``Product'' variant
       }
  \end{subfigure}%

   \caption{\label{exp:maps_an_comparison} Maps of elections with 10 candidates and 50 voters. Each point represents a single election. The darker the point is, the more voters agree on a certain set of candidates being antagonism. On each map, ID label marks the identity election, and AN label marks the antagonism election. }
\end{figure*}

\begin{table}[t]
  \centering
{
  \begin{tabular}{lc}
    \toprule
    Statistical Culture & Number of Elections \\
    
    \midrule
    Impartial Culture           & 20 \\
    \midrule
        SP by Conitzer  & 20 \\
    SP by Walsh      & 20 \\
    SPOC                 & 20 \\
    Single-Crossing           & 20 \\
    \midrule
    1D        & 20 \\
    2D         & 20 \\
    3D         & 20 \\
    Circle         & 20 \\
    \midrule
    GS Balanced    & 20 \\
    GS Caterpillar    & 20 \\
        \midrule
        Urn       & 60 \\  
    Norm-Mallows    & 60 \\

    \midrule
    Compass ($\ID$,~$\AN$,~$\UN$,~$\ST$) & 4 \\
    \bottomrule
  \end{tabular} 
}
    \caption{List of selected statistical cultures and numbers of elections from these cultures accordingly.}
\end{table}\label{tab:testbed}

  


\newcommand{\nbarC}[1]{\tikz{
    \fill[green!80!black!16] (0,0) rectangle (#1*0.5mm,8pt);
    \node[inner sep=0pt, anchor=south west] at (0,0) {#1};}
}
\newcommand{\nbarID}[1]{\tikz{
    \fill[blue!80!black!16] (0,0) rectangle (#1*0.5mm,8pt);
    \node[inner sep=0pt, anchor=south west] at (0,0) {#1};}
}
\newcommand{\nbarAN}[1]{\tikz{
    \fill[red!80!black!16] (0,0) rectangle (#1*0.5mm,8pt);
    \node[inner sep=0pt, anchor=south west] at (0,0) {#1};}
}

\begin{table}[t]
	\begin{center}
		
	\scalebox{1}{	\begin{tabular}{l l c}
			            \toprule
			Statistical Culture & Average Value & Std. \\ 
       \midrule

            Impartial Culture & \nbarC{17.0} & 1.22 \\
            \midrule
SP by Conitzer & \nbarC{41.7} & 1.55 \\
SP by Walsh & \nbarC{34.7} & 2.39 \\
SPOC & \nbarC{33.2} & 2.34 \\
Single-Crossing & \nbarC{33.5} & 3.75 \\
      \midrule
1D & \nbarC{48.3} & 2.12 \\
2D & \nbarC{35.75} & 5.49 \\
3D & \nbarC{32.15} & 6.17 \\
Circle & \nbarC{44.85} & 4.84 \\
      \midrule
GS Balanced & \nbarC{50.0} & 0.0 \\
GS Caterpillar & \nbarC{50.0} & 0.0 \\
      \midrule
*Urn & \nbarC{35.55} & 10.13 \\
*Norm-Mallows & \nbarC{27.25} & 10.73 \\

			\bottomrule
		\end{tabular}}
	\end{center}
	
	\caption{$\maxclone(E,2)$ for elections with 10 candidates and 50 voters. (*Note that average values for Urn and Mallows models are not very meaningful because these models are parameterized).}
	
\end{table}\label{fig:overview_ihc_max_2}

\begin{table}[t]
	\begin{center}
		
	\scalebox{1}{	\begin{tabular}{l l c}
			         \toprule
			Statistical Culture & Average Value & Std. \\ 

               \midrule
   Impartial Culture & \nbarID{4.85} & 0.65 \\
   \midrule
SP by Conitzer & \nbarID{31.5} & 1.16 \\
SP by Walsh & \nbarID{39.8} & 2.06 \\
SPOC & \nbarID{16.0} & 1.92 \\
Single-Crossing & \nbarID{30.1} & 4.49 \\
\midrule
1D & \nbarID{34.65} & 3.21 \\
2D & \nbarID{19.3} & 3.74 \\
3D & \nbarID{14.6} & 3.88 \\
Circle & \nbarID{21.65} & 3.13 \\
\midrule
GS Balanced & \nbarID{8.95} & 0.92 \\
GS Caterpillar & \nbarID{9.25} & 0.94 \\
\midrule
*Urn & \nbarID{27.32} & 12.49 \\
*Norm-Mallows & \nbarID{20.05} & 15.77 \\
			\bottomrule
		\end{tabular}}
	\end{center}\label{fig:overview_ihi_max_5}
	
	\caption{$\maxid(E,5)$ for elections with 10 candidates and 50 voters. (*Note that average values for Urn and Mallows models are not very meaningful because these models are parameterized).}
	
\end{table}

\begin{table}[t]
	\begin{center}
		
	\scalebox{1}{	\begin{tabular}{l l c}
			 \toprule
			Statistical Culture & Average Value & Std. \\ 
       \midrule
            Impartial Culture & \nbarAN{4.8} & 0.98 \\
            \midrule
SP by Conitzer & \nbarAN{32.5} & 3.63 \\
SP by Walsh & \nbarAN{12.2} & 3.57 \\
SPOC & \nbarAN{23.9} & 2.41 \\
Single-Crossing & \nbarAN{17.9} & 5.81 \\
\midrule
1D & \nbarAN{31.0} & 5.04 \\
2D & \nbarAN{16.7} & 4.11 \\
3D & \nbarAN{11.7} & 3.05 \\
Circle & \nbarAN{32.0} & 4.0 \\
\midrule
GS Balanced & \nbarAN{11.8} & 1.4 \\
GS Caterpillar & \nbarAN{12.4} & 1.36 \\
\midrule
*Urn & \nbarAN{12.73} & 8.68 \\
*Norm-Mallows & \nbarAN{2.87} & 1.98 \\
			\bottomrule
		\end{tabular}}
	\end{center}
	
	\caption{$\maxan(E,5)$ for elections with 10 candidates and 50 voters. (*Note that average values for Urn and Mallows models are not very meaningful because these models are parametrized).}
	\label{fig:overview_iha_max_5}
\end{table}


\begin{table}[t]
  \centering
{
  \begin{tabular}{l | c c | c c c}

    \toprule
    Variant & AN & ID & Default & Sum & Product\\
    
    \midrule
      Default   & -0.845 & 0.352  & 1 & 0.433 & 0.929 \\
      Sum  & -0.183 & -0.598  & 0.433 & 1 & 0.589 \\
      Product      & -0.739  & 0.164  & 0.929 & 0.589 & 1 \\
   
    \bottomrule
  \end{tabular} }
    \caption{ PCC between three different variants of antagonism, and the distances from $AN$ and from $ID$.}
    \label{tab:pcc_an_comparison}
\end{table}

\end{document}